\documentclass[journal]{IEEEtran}
\ifCLASSINFOpdf
\else
\fi

\usepackage{amsmath}
\usepackage{amssymb}
\usepackage{mathtools}
\usepackage{amsfonts}

\usepackage{amsthm}
\usepackage{overpic}
\usepackage{lipsum}

  {
      \theoremstyle{plain}
      \newtheorem{assumption}{Assumption}
  }

\usepackage[utf8]{inputenc}
\usepackage[T1]{fontenc}
\usepackage{url}
\usepackage{ifthen}
\usepackage{cite}
\usepackage{placeins}

\usepackage{graphicx,epstopdf,amssymb,amsthm}
\usepackage{ulem,color}
\usepackage{epstopdf}
\usepackage[utf8]{inputenc}
\usepackage[english]{babel}
\usepackage{caption}

\usepackage{enumitem}   
\usepackage{mathtools}

\DeclareCaptionLabelFormat{lc}{\MakeLowercase{#1}~#2}
\newtheorem{theorem}{Theorem}
\newtheorem{corollary}{Corollary}
\newtheorem{remark}{Remark}
\newtheorem{lemma}{Lemma}

\newtheorem{definition}{Definition}
\newtheorem{proposition}{Proposition}
\newtheorem{example}{Example}[]

\usepackage{xcolor}

\begin{document}
%
\title{
Analysis, Control, and State Estimation for the Networked Competitive Multi-Virus SIR Model
}
%
%
%

\author{Ciyuan Zhang, Sebin Gracy, Tamer Ba\c sar, and Philip E. Par\'e*
\thanks{An earlier version of part of this work was presented at the IFAC Conference on Networked Systems (Necsys22), Zurich, Switzerland~\cite{zhang2022networked}.}
\thanks{*Ciyuan Zhang and~Philip E. Par\'e are with the Elmore Family School of Electrical and Computer Engineering at Purdue University, West Lafayette, IN, 47907, USA. Emails: \{zhan3375, philpare\}@purdue.edu. Sebin Gracy is with the Electrical and Computer Engineering Department at Rice University, Houston, TX, 77005, USA. Email: sebin.gracy@rice.edu. Tamer Ba\c sar is with the Coordinated Science Laboratory at the University of Illinois at Urbana-Champaign, Urbana, IL, 61801, USA. Email: basar1@illinois.edu.
Research supported in part by the Rice Academy of Fellows, Rice University, and the National Science Foundation, grants NSF-ECCS \#2032258 and NSF-ECCS \#2032321.}}
\maketitle

\begin{abstract}
This paper proposes a novel discrete-time multi-virus  susceptible-infected-recovered (SIR) model that captures the spread of competing 
epidemics over a population network. First, we provide sufficient conditions for the infection level of all the viruses over the networked model to converge to zero in exponential time. Second, we propose an observation model which captures the summation of all the viruses' infection levels in each node, which represents the individuals who are infected by different viruses but share similar symptoms. Third, we present a sufficient condition for the model to be strongly locally observable, assuming that the network has only infected or recovered individuals. Fourth, we propose a Luenberger observer for estimating the states of our system.
We 
prove that the estimation error of our proposed estimator converges to zero asymptotically with the observer gain. 
Finally, we present a distributed feedback controller which guarantees that each virus dies out at an exponential rate.  We then show via simulations that the estimation error of the Luenberger observer converges to zero before the viruses die out. 

\end{abstract}


%
\IEEEpeerreviewmaketitle

\section{Introduction}
The history of human civilization has been a narrative of undergoing, battling, and outmatching various pandemics~\cite{benedictow2004black, cheng2007happened, johnson2002updating}. 
The devastation that epidemics leave in their wake is well-known; as examples, 
the $2009$ swine flu pandemic, which was caused by the H1N1 influenza virus resulted in $60.8$ million cases worldwide~\cite{whoswine}, and the more recent SARS-CoV-2 virus infected 642 million individuals globally~\cite{whoCoronavirus}. Given the detrimental impacts that epidemics have on society,
research 
on the spread of multiple diseases has attracted attention from several research communities. In fact, research activity 
has grown 
massively 
with the investigation of each epidemic.
%
Various infection models have been proposed and studied in the literature to capture epidemic processes, based on the characteristics of individual pathogens, some of the most basic ones 
being the susceptible-infected-susceptible (SIS), susceptible-infected-removed (SIR), and susceptible-infected-removed-susceptible (SIRS) models~\cite{mei2017dynamics}. 
The SIR model is a simple yet classic epidemic model for diseases in which hosts recover with permanent or close to permanent immunity following infection~\cite{weiss2013sir}. Diseases belonging to this category include a wide range of airborne diseases: Spanish Flu, SARS~\cite{colizza2007predictability}, MERS~\cite{chang2017estimation}, Influenza~\cite{osthus2017forecasting}, H1N1~\cite{coburn2009modeling}, Ebola~\cite{berge2017simple}, and COVID-19~\cite{calafiore2020modified, chen2020time}, to cite a few.

To make the situation more complicated (and realistic),
it is not unusual for
multiple
viruses/strains to be simultaneously 
active in
a community. 
When there are multiple viruses circulating, those could either be cooperative, in which case infection increases the likelihood of simultaneous infection with another virus \cite{gracy2022modeling}, or competitive, in which case infection with one virus precludes the possibility of simultaneous infection with another virus, \cite{liu2019analysis,ye2021convergence}. 
Inspired by the competition of different virus strains in population networks~\cite{pepin2008asymmetric, nowak1991evolution, poland1996two}, 
we investigate in this paper the spread of competitive viruses over a population network. Note that the network model can efficiently capture the spread of epidemics over large populations as compared to single-population compartmental models~\cite{van2009virus}.
An extensive amount of effort has been expended on the study of
multi-virus models~\cite{pare2017multi, prakash2012winner, pare2020analysis, sahneh2014competitive, santos2015bi, liu2016analysis, pare2021multi, janson2020networked, ye2021convergence}, all of which focus on the competing SIS networked virus model, whereas
a susceptible-infected-water-susceptible (SIWS) multi-virus model with a shared resource was investigated in \cite{janson2020networked}. 
The competing networked epidemic models can be applied to many fields other than epidemiology, such as modeling the spread of competitive opinions spread over social networks~\cite{sahneh2014competitive}, rival merchandise's sales in a market~\cite{prakash2012winner}, and competing US Department of Agriculture farm subsidy programs~\cite{pare2020analysis}.
\par The single virus SIR epidemic networked model has been studied extensively, e.g.,~\cite{zhang2021estimation, ito2022strict, she2021peak, she2022learning, smith2022convex}. 
The competing SIR epidemic model has 
been investigated in
~\cite{ignatov2022two}. 
The focus in~\cite{ignatov2022two}, however, was on a 
competing bi-virus (two viruses) SIR epidemic model; it does not account for the networked case nor does it consider 
an arbitrary but finite number of viruses. 
To the best of our knowledge, a networked multi-competitive SIR model that accounts for the presence of an arbitrary but finite number of viruses has
not yet been studied in the literature. 
Thus, in this work, 
we propose a discrete-time multi-competitive networked SIR model. 
The multi-virus model captures the presence and spread of multiple viruses/variants over a population, 
such as different variants of 
SARS-CoV-2 virus, 
and could also be utilized to represent different behaviors of strains of viruses that compete with each other
~\cite{lopez2021effectiveness}.


Beyond the modeling and analysis of 
epidemic models, 
monitoring of epidemics and estimation of infection levels in the population 
have been some of the fundamental quests in 
the research on competing contagions.
Given that the SARS-CoV-2 pandemic has provided us with an enormous amount of data, the question of how to accurately infer the infection levels 
with respect to various strains of a virus (Delta, Omicron, etc.) and/or with respect to different viruses (influenza, SARS-CoV-2) 
have been pursued by
the research community~\cite{barmparis2020estimating, russell2020estimating, meyerowitz2020systematic}.
Note that certain infectious diseases, 
in particular Delta and Omicron variants of the SARS-CoV-2 virus \cite{antonelli2022risk, rader2022use}
in spite of being competing, exhibit similar symptoms. Consequently, it is quite likely that an individual suffering from the Omicron variant could be diagnosed as having Delta, and vice-versa. Thus,  certain competing SIR epidemics can affect the measurement of the cases corresponding to each epidemic and, therefore,  pose difficulties for the estimation of the states of 
each epidemic.
When testing facilities are limited, as was witnessed at the beginning of the SARS-CoV-2 pandemic and at various peaks of different waves~\cite{perks2021covid}, it is even more difficult to, assuming an individual is infected, accurately infer which virus is it that said individual is infected with\cite{ma2021metagenomic, belongia2020covid}.
In this paper, we consider the following research question: assuming that there are multiple competing viruses spreading over a network, given measurements of individuals who exhibit symptoms corresponding to at least one (but possibly more) of the viruses, can we accurately estimate the infection level with respect to each virus?
State estimation of non-linear systems having a structure similar to SIR epidemic models has been studied in~\cite{niazi2022observer, bara2005observer, arcak2001nonlinear}. However, scenarios where the observation is an accumulation of the system states with a dimension lower than that of the system states have not been investigated in the literature. 

In \cite{mitra2018distributed}, the authors thoroughly investigated a distributed observer for discrete-time LTI networked systems, where the observation model 
collects measurements 
of the system states from a 
node and its neighbors. The paper focused on the LTI systems which 
do not capture the nonlinear feature of the epidemic systems; it did not consider the scenario when the dimension of the observation vector is lower than the dimension of the state space. 
However, 
the distributed approach for observer design in \cite{mitra2018distributed} has inspired us to design a distributed estimator which considers the measurements of the local node and the inferred systems states of the 
neighboring nodes
in the network for estimation. Different from~\cite{mitra2018distributed}, we assume that each node (corresponds to a subpopulation in our case) obtains completely accurate information regarding the values that the states of its neighbors take and that, furthermore, there is no time delay.

Note that the discussion so far has focused on the stability analysis, observability, and design of the observer for the model that we have proposed. Yet another challenging problem that policymakers are confronted with during the course of a pandemic is the design of effective eradication/mitigation strategies. 
As we will show later in the paper, all the viruses in the competing SIR networked model die out eventually regardless of the system parameters. However, the exponential decay of the viruses is essential for the decision-makers as it guarantees that fewer individuals become infected by the virus over the course of the outbreak compared to non-exponential convergence~\cite{zhang2021estimation}. 
To the best of our knowledge, distributed feedback control for competing SIR networked epidemic models has remained an open problem- the present paper aims to address this gap as well.

Earlier versions of some of the results in this paper were presented in~\cite{zhang2022networked}.
This paper substantially expands upon 
the conference version 
by:
\begin{itemize}
    \item 
    deriving a result
    guaranteeing that the estimation error converges to zero asymptotically (see Theorem~\ref{thm:Error_GAS});
    \item proposing a distributed mitigation strategy that eradicates all the viruses at an exponential rate, (see Theorem~\ref{thm:control_feedback}); and
    \item presenting additional simulations which corroborate the theoretical results of the paper and show how to choose the feasible observer gain such that the estimated system states are well-defined.
\end{itemize}

\subsection{Paper Contributions}

Through this paper, we make the following contributions:
\begin{enumerate}[label=\roman*)]
    \item We propose 
    a discrete-time competitive multi-virus networked SIR model; see Eq.~\eqref{eq:dt_SIR}.
    \item We 
    provide sufficient conditions for the infection level of each virus to converge to zero in exponential time; see Theorem~\ref{thm:GES} 
    and Proposition~\ref{prop:rho_tilde}.
    \item We specify a necessary and sufficient condition that guarantees that the system is strongly locally observable at a certain system state; see Theorem~\ref{thm:Observable_zero}.
    \item We also propose a distributed state estimator and study how to design the observer gain such that the system state estimation errors converge to zero asymptotically; see 
    Theorem~\ref{thm:Error_GAS}.
    \item We propose a distributed mitigation algorithm that ensures that the infection level of each virus converges to zero within at least exponential time; see Theorem~\ref{thm:control_feedback}.
\end{enumerate}

\subsection{Paper Outline}

The paper is organized
as follows: Section~\ref{sec:background} proposes a novel discrete-time competitive multi-virus networked SIR model, introduces the research questions addressed in the paper, and provides preliminary results. Section~\ref{sec:stability} 
provides two sufficient conditions for eradication of a virus in exponential time.
Section~\ref{sec:observation} examines the observation model, specifies a 
sufficient condition for the system to be strongly locally observable at a system state (namely, where every node is either in the infected or in the 
recovered
compartment), proposes a distributed Luenberger 
estimator
which estimates the states of the aforementioned model,
and determines 
a condition which ensures that the estimation error of the proposed observer converges to zero asymptotically.
Section~\ref{sec:mitigation} proposes a distributed feedback controller 
that guarantees that each virus dies out at an exponential rate.
In Section~\ref{sec:simulation}, we utilize simulations to illustrate the results from Sections~\ref{sec:stability},~\ref{sec:observation}, and \ref{sec:mitigation}. Finally, Section~\ref{sec:conclusion} summarizes the main 
findings
of the paper. 

\subsection{Notation}

We denote the set of real numbers and the set of non-negative integers by $\mathbb{R}$ and $\mathbb{Z}_{\geq 0}$, respectively. For any positive integer $n$, we have $[n]:=\{1,2,...,n \}$. The spectral radius and an eigenvalue of a matrix $A \in \mathbb{R}^{n\times n}$ are denoted by $ \rho(A)$ and $\lambda(A)$, respectively. A diagonal matrix is denoted by diag$(\cdot)$. The transpose of a vector $x\in \mathbb{R}^n$ is $x^\top$. The Euclidean norm is denoted by $\lVert \cdot \rVert$. We use $I$ to denote the identity matrix. We use $\mathbf{0}$ and $\mathbf{1}$ to denote the vectors whose entries are all $0$ and $1$, respectively, where the dimensions of the vectors are determined by context. Given a matrix $A$, $A \succ 0$ 
indicates that $A$ is positive definite, 
whereas $A \prec 0$ 
indicates that A is negative definite, and $A^{-1}$ represents the inverse of matrix $A$.
Let $G =(\mathbb{V},\mathbb{E})$ denote a graph or network where $\mathbb{V} = \{ v_1, v_2,..., v_n\}$ is the set of nodes, and $\mathbb{E} \subseteq \mathbb{V}\times \mathbb{V} $ is the set of edges.


 




\section{Background}\label{sec:background}
In this section, we present our system model, recall some preliminary results that will be used in the sequel, 
and 
formally state the research questions that the paper addresses. 


\subsection{System Model}

\begin{figure}
\centering
\includegraphics[width=.35\textwidth]{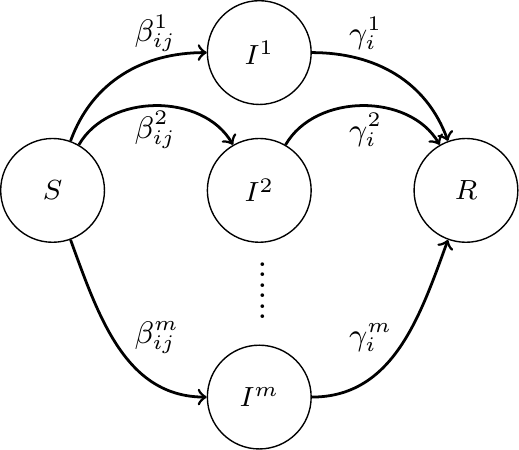}
\caption{The competing SIR networked model. Each node in the network can only be in one of $(m+2)$ states: $S$, $I^k$, $R$, where $k \in[m]$.}
\label{fig:SI^nR}
\end{figure}

Consider a network of $n \geq 2$ nodes in which $m$ ($m \geq 2$), viruses compete to infect the nodes. 
A node here represents a well-mixed population (hereafter, referred to as a subpopulation) of individuals with a large and constant size. A node could be in one of $m+2$ mutually exclusive compartments: susceptible, infected with virus~$k$, 
$k=[m]$,
or recovered. We say that a node is healthy if all individuals within it are healthy; otherwise, we say that it is infected. An individual in the susceptible compartment transitions to the "infected with virus~$k$" (for $k \in [m]$) compartment depending on its infection rate with respect to virus~$k$, namely $\beta_i^k$. An individual in the "infected with virus~$k$" compartment transitions to the recovered compartment depending on its recovery rate with respect to virus~$k$, namely $\gamma_i^k$.

We use an $m$-layer graph $G$ to capture 
the spread of $m$-competing viruses. 
The vertices of $G$ correspond to the population nodes. The contact  graph
that represents the pathways through which virus~$k$, for each $k \in [m]$, spreads in the population
 is denoted by the $k^\textrm{th}$ layer. 
In particular, there exists a directed edge from node $j$ to node $i$ in layer $k$ if, assuming an individual in population $j$ is infected with virus~$k$, then said individual can infect at least one  healthy individual in node~$i$. The edge set corresponding to the $k^\textrm{th}$ layer of $G$ is denoted by $E^k$, while $ A^k$ denotes the weighted adjacency matrix 
 (where $a_{ij}^k \geq 0$).
 Note that $(i,j)\in E^k$ if, and only if, $a_{ji}^k\neq 0$. We denote by $s_i$ and $r_i$, respectively, the susceptible and recovered proportions of subpopulation $i$. We use $x^k_i[t]$, where $k \in [m]$, to denote the fraction of individuals at node $i$ infected with virus $k$ at time instant $t$. A pictorial depiction of this
model is given in Figure~\ref{fig:SI^nR}. 
In continuous time, the dynamics of the $i$-th node can be written as follows:
\begin{subequations}
\IEEEyesnumber\label{eq:ct_SIR} 
\begin{align}
 \dot{s}_i&= -(1-x_i^1 - \cdots-x_i^m-r_i) \sum_{k=1}^m\sum_{j=1}^n \beta^k_{ij}x^k_j, \label{eq:ct_SIRsub1}\\
\dot{x}^k_i &= (1-x_i^1 - \cdots-x_i^m-r_i) \sum_{j=1}^n \beta^k_{ij}x^k_j -\gamma^k_i x^k_i,  \label{eq:ct_SIRsub2}
\\
\dot{r}_i &=  \sum_{k=1}^m\gamma^k_i x^k_i, \;\ \forall i \in [n].
\label{eq:ct_SIRsub3}
\end{align}
\end{subequations}
Note that while an epidemic process evolves in continuous time, the data regarding the evolution of an epidemic are compiled on a daily basis \cite{whoCoronavirus,snow1855mode} or on a weekly basis \cite{whoEbola}. Such a sampling of the system behavior  motivates the use of a discrete-time multi-competitive networked SIR model.
By using the Euler method~\cite{atkinson2008introduction}, we derive the discrete-time dynamics 
of the SIR networked epidemic model at node $i$ as:
\begin{subequations}
\IEEEyesnumber\label{eq:dt_SIR} 
\begin{align}
  &s_i[t+1] = s_i[t] -h\Bigg(s_i[t]\sum_{k=1}^m\sum_{j=1}^n \beta^k_{ij} x^k_j[t]\Bigg), \label{eq:dt_SIRsub1}\\
 &x^k_i[t+1] = x^k_i[t]+ h\Bigg(s_i[t]\sum_{j=1}^n \beta^k_{ij} x^k_j[t]-\gamma^k_i x^k_i[t]\Bigg),  \label{eq:dt_SIRsub2}
\\
 &r_i[t+1] = r_i[t] +h\sum_{k=1}^m \gamma^k_i x^k_i[t], \label{eq:dt_SIRsub3}
\end{align}
\end{subequations}
where $h>0$ is the sampling parameter, 
$t$ is the time index, and $k\in [m]$ indicates the $k$-th virus. Notice that $s_i[t] +x_i^1[t] + \cdots+x_i^m[t]+r_i[t]= 1$, capturing the fact that in the competing virus scenario, all the viruses are mutually exclusive. 
We now rewrite~\eqref{eq:dt_SIRsub2} in compact form as:
\begin{equation}\label{eq:SIR_dynamic}
    x^k[t+1] =x^k[t] +h \{S[t]B^k - \Gamma^k\}x^k[t],
\end{equation}
where $S[t] = $ diag$(s_i[t])$, 
$B^k$ is a matrix with ($i,j$)-th entry $\beta^k_{ij}$, and $\Gamma^k[t] =$ diag$(\gamma^k_i)$.

We now introduce the following assumptions to ensure that the model in~\eqref{eq:dt_SIR} is well-defined. 

\begin{assumption}\label{assume:one}
For all $i \in [n]$ and $k \in [m]$, we have $s_i[0], x^k_i[0], r_i[0] \in [0,1]$ and $s_i[0]+\sum_{k=1}^m x^k_i[0]+r_i[0] =1$.
\end{assumption}

\begin{assumption}\label{assume:two}
For all $i \in [n]$, and $k \in [m]$, we have $\beta^k_{ij} \geq 0, \gamma^k_i> 0$.
\end{assumption}

\begin{assumption}\label{assume:three}
For all $i \in [n]$,  and $k \in [m]$, we have $h\sum_{k=1}^m \sum_{j=1}^n \beta^k_{ij} \leq 1$ and $ h\sum_{k=1}^m \gamma^k_i\leq 1$.
\end{assumption}

Assumptions~\ref{assume:one} and \ref{assume:two} can be interpreted as the initial proportion of susceptible, infected, and recovered individuals, all 
lying
in the interval 
$[0,1]$, and 
that the healing rates are always positive, which are both reasonable~\cite{mei2017dynamics, brauer2019mathematical, she2021peak}. Assumption~\ref{assume:three} ensures that the sampling rate is frequent enough for the states of the model to remain well-defined.



Motivated by the fact that different infectious diseases can demonstrate similar symptoms over the hosts~\cite{ma2021metagenomic, belongia2020covid}, 
we build an observation model which produces the output as the aggregated proportion of individuals who show flu-like symptoms from infection of all viruses. The observation model is written as  (where we repeat \eqref{eq:dt_SIRsub2} for convenience):
\begin{subequations}
\IEEEyesnumber\label{eq:SIR_observation} 
\begin{align}
  x^k_i[t+1] &= x^k_i[t] +h\bigg\{s_i[t]\sum_{j=1}^n \beta^k_{ij} x^k_j[t]-\gamma^k_i x^k_i[t]\bigg\}, \label{eq:SIR_observationsub1}\\
y_i[t] & = \sum_{k=1}^m c_i^k x_i^k[t],\label{eq:SIR_observationsub2}
\end{align}
\end{subequations}
where $c_i^k$ is the measurement coefficient.
\begin{assumption}\label{assume:four}
The coefficient $c_i^k\in (0,1]$ for all $i\in [n], k\in[m]$.
\end{assumption}

\begin{remark}
The coefficient $c_i^k$ 
from Eq.~\eqref{eq:SIR_observationsub2} can capture the probability of showing symptoms from the $k$-th virus at subpopulation~$i$. Therefore, $1-c_i^k$ captures the probability of individuals infected by the $k$-th virus in subpopulation~$i$
being asymptomatic. 
The probability of exhibiting symptoms corresponding to different viruses has been studied in, among others,
~\cite{teunis2008norwalk, panovska2020determining}.
The coefficient $c_i^k $ can also represent how each subpopulation $i$ defines and measures the cases based on the symptoms of each virus $k$. For example, the symptoms of the SARS-CoV-2 virus can include but are not limited to fever, muscle aches, cough, runny nose, headaches, and fatigue.
\end{remark}

\begin{remark}
Given a subpopulation $i$, Eq.~\eqref{eq:SIR_observationsub2} captures the total proportion of symptomatic patients at subpopulation $i$. This value, namely $y_i[t]$, is a useful indicator for decision-makers to formulate policies regarding the adequacy of local healthcare facilities and the availability of virus-combating resources.
\end{remark}
\noindent We next present the following results for the system model under Assumptions~\ref{assume:one}-\ref{assume:three}. 

\begin{lemma}\label{lemma:one}
Under Assumptions~\ref{assume:one}, \ref{assume:two}, and \ref{assume:three}
, for all $i \in[n]$ and $t \in \mathbb{Z}_{\geq0}$,
\begin{enumerate}
    \item $s_i[t], x^k_i[t], r_i[t] \in [0,1]$, for all $k \in [m]$, and $s_i[t]+\sum_{k=1}^m x^{k}_i[t]+r_i[t] =1$, and 
    \item $s_i[t+1] \leq s_i[t]$.
\end{enumerate}
\end{lemma}
\begin{proof}
Proof of statement  1): We prove this result by induction. 
  
\noindent \textit{Base Case:}
By the assumptions made, $s_i[0], x^k_i[0], r_i[0] \in [0,1]$, $s_i[0]+\sum_{k=1}^m x^k_i[0]+r_i[0] =1$ for all $k\in[m], i\in[n]$. From Assumptions~\ref{assume:one}-\ref{assume:three}, we know that $s_i[0]\geq 0$ and $1-h\sum_k^m \sum_j^n \beta_{ij}^k x_j^k[0] \geq0$, and hence $s_i[1] = s_i[0](1-h\sum_k^m \sum_j^n \beta_{ij}^k x_j^k[0])\geq0$. Since $-h(s_i[t]\sum_k^m \sum_j^n \beta_{ij}^kx_j[0])\leq 0$, we obtain that $s_i[1]\leq s_i[0]\leq 1$. 
We can also acquire that $x_i^k[1]\geq x_i^k[0](1-h\gamma_i^k)\geq 0$ and $x_i^k[1]\leq x_i^k[0]+hs_i[0]\sum_j^n \beta_{ij}^k x_j^k[0]\leq x_i^k[0]+s_i[0]\leq 1$. Ultimately, we have $r_i[1]\geq r_i[0]\geq0$ and $r_i[1]\leq r_i[0]+ \sum_k^m x_i^k[0]\leq 1$. Summing up Eqs.~\eqref{eq:dt_SIRsub1}-\eqref{eq:dt_SIRsub3}, we obtain that $s_i[1]+\sum_{k=1}^m x^{k}_i[1]+r_i[1] =s_i[0]+\sum_{k=1}^m x^{k}_i[0]+r_i[0] =1$.

\noindent \textit{Inductive Step:} We assume for some arbitrary $t$ that the following holds: $s_i[t], x^k_i[t], r_i[t] \in [0,1]$, for all $k \in [m]$ and $s_i[t]+\sum_{k=1}^m x^{k}_i[t]+r_i[t] =1$. By repeating the same steps from the \textit{Base Case} except replacing $0$ and $1$ with $k$ and $k+1$, we can write that $s_i[t+1], x^k_i[t+1], r_i[t+1] \in [0,1]$, for all $k \in [m]$ and $s_i[t+1]+\sum_{k=1}^m x^{k}_i[t+1]+r_i[t+1] =1$. Therefore, by induction, we can prove that $s_i[t], x^k_i[t], r_i[t] \in [0,1]$, for all $k \in [m]$ and $s_i[t]+\sum_{k=1}^m x^{k}_i[t]+r_i[t] =1$ for all $i \in[n]$ and $t\in  \mathbb{Z}_{\geq0}$.

Proof of statement 2): From 1) and Assumption~\ref{assume:two} we know that $-h\Big(s_i[t]\sum_{k=1}^m\sum_{j=1}^n \beta^k_{ij} x^k_j[t]\Big)\leq 0$ for all $t\in  \mathbb{Z}_{\geq0}$. Thus, we have $s_i[t+1]\leq s_t[t]$ for all $i\in [n]$ and $t\in  \mathbb{Z}_{\geq0}$. 

\end{proof}


\subsection{Preliminaries}
In this subsection, we recall certain preliminary results from  non-linear systems theory that will help in the development of the main results of the paper.
\par Consider a system described as follows:
\begin{subequations}
\IEEEyesnumber\label{eq:dt_prelim} 
\begin{align}
  x[t+1] &= f(t,x[t]), \label{eq:dt_prelimsub1}\\
 y[t] & = g(x[t]). \label{eq:dt_prelimsub2}
\end{align}
\end{subequations}

\begin{definition}\label{def:GES}
An equilibrium point of \eqref{eq:dt_prelimsub1} is GES if there exist positive constants $\alpha$ and $\omega$, with $0\leq \omega <1$, such that
\begin{equation}
    \lVert x[t] \rVert \leq \alpha \lVert x[t_0] \rVert \omega^{(t-t_0)}, \forall t\geq t_0 \geq 0, \forall x[t_0] \in \mathbb{R}^n.
\end{equation}
\end{definition}
\begin{lemma}\label{lemma:GES}
\cite[Theorem 28]{vidyasagar2002nonlinear}
Suppose that there exist a function $V: \mathbb{Z}_+ \times \mathbb{R}^n \rightarrow \mathbb{R}$, and constants $a,b,c>0$ and $p>1$ such that $a\lVert x \rVert^p \leq V(t,x) \leq b \lVert x \rVert^p$, $\Delta V(t,x):= V(x[t+1])-V(x[t]) \leq -c\lVert x \rVert^p, \forall t \in \mathbb{Z}_{\geq0}$. Then $\forall x (t_0)\in \mathbb{R}^n$, $x=\mathbf{0}$ is the globally exponential stable equilibrium of \eqref{eq:dt_prelimsub1}. 
\end{lemma}
\begin{lemma}\label{lemma:rate_GES}
\cite[Theorem 23.3]{rugh1996linear} Under the conditions of Lemma~\ref{lemma:GES}, convergence to 
the origin 
has an exponential rate of at least $\sqrt{1-(c/b)} \in [0,1)$, where $b$ and $c$ are as defined in Lemma~\ref{lemma:GES}.
\end{lemma}
\begin{lemma}\label{lemma:diagonal}
\cite[Proposition 2]{rantzer2011distributed}
Suppose that $M$ is a nonnegative matrix which satisfies $\rho(M)<1$. Then, there exists a diagonal matrix $P \succ 0$ such that $M^\top P M -P \prec 0$. 
\end{lemma}
 
\begin{definition}
The system in Eq.~\eqref{eq:SIR_observation} is strongly locally observable at 
$s[t]$
if we are able to recover 
$x_i^k [t]$ for all $i \in [n], k\in [m]$
through the output in the duration of $[t,t+m-1]$.
\end{definition}
 
\begin{lemma}\cite{sontag1979observability, nijmeijer1982observability}\label{lemma:locally_observable}
The system in~\eqref{eq:dt_prelim} is strongly locally 
observable at $x[t]$ 
if and only if the map $x[t] \rightarrow (g(x[t]),g(f^1(x[t])),\cdots, g(f^{n-1}(x[t])))$ 
is injective, where $n$ is the dimension of $x[t]$.
\end{lemma}



 

\subsection{Problem Formulation}

With the system model set up as above, we now introduce the problems 
investigated in this paper. 


\begin{enumerate}[label=(\roman*)]
    \item \label{itm:first} For the system with dynamics given in~\eqref{eq:SIR_dynamic}, provide a sufficient condition which ensures that $x^k[t]$ for some $k\in[m]$ converges to the eradicated state, namely $x^k=\mathbf{0}$, in exponential time.
    \item What is the rate of convergence for the sequence $x^k[t], k\in[m]$ (converging to $\mathbf{0}$ exponentially)?
    \item \label{itm:third} 
    Provided that the infection rate matrix $B^k$, healing rate matrix $\Gamma^k$, the observation $y_i[t]$, true susceptible states $s_i[t]$, and the measurement coefficient $c_i^k$, are known, for all $i\in [n], k\in [m]$, under what conditions are the infection levels of each virus $x_i^k[t]$, for all $i\in [n], k \in [m] $ strongly locally observable, at $s_i[t]=0, \forall i \in [n]$?
    \item Given the system parameters $\beta^k_{ij}, \gamma_i^k$ for all $i,j \in[n], k \in [m]$, and the local aggregated observation, $y_i[t]$ for all $i \in [n]$, 
construct a distributed Luenberger observer 
    which delivers
    the system states $\hat{s}_i[t], \hat{x}^k_{i}[t], \hat{r}_i[t]$ for all $i\in [n], k \in[m]$?
    \item Given the system parameters $\beta^k_{ij}, \gamma_i^k$ for all $i,j \in[n], k \in [m]$, how do we find the gain of the distributed observer
    such that the estimation error $x^k[t]-\hat{x}^k[t]$ converges to zero asymptotically for all $k \in[m]$?
    \item 
Design a distributed feedback controller which eradicates all viruses at an exponential rate.
\end{enumerate}

\section{Healthy State Analysis} \label{sec:stability}

In this section, we identify multiple sufficient conditions which guarantee
that each virus $k$ converges to zero exponentially fast, and provide the associated rates of convergence for each virus. Similar to the 
standard single virus 
SIR 
networked 
model, the multi-competitive SIR networked model also converges to a healthy state regardless of i) the values that the system parameters take, and ii) initial conditions. However, it is important to study the exponential convergence since it guarantees that the viruses die out at a faster rate and, as  a consequence, fewer individuals become infected over the course of the outbreak~\cite{zhang2021estimation}.


Let
\begin{align}
     M^k &:= I-h\Gamma^k +hB^k, \label{eq:M} \\
     \Tilde{M}^k[t]  &:= I+ h\{S[t]B^k - \Gamma^k\},
    \label{eq:Mhat}
\end{align}
and note that $\Tilde{M}^k $ is the state 
matrix of system~\eqref{eq:SIR_dynamic}:
\begin{equation}\label{eq:M_M_hat}
    \Tilde{M}^k[t] = M^k-h(I-S[t])B^k.
\end{equation}
We first present a sufficient condition, in terms of $M^k$, for the infection level with respect to each of the viruses to converge to zero exponentially.


\begin{theorem}\label{thm:GES}
Consider system~\eqref{eq:dt_SIR} under Assumptions \ref{assume:one}-\ref{assume:three}. If $\rho(M^k)<1$, then the $k$-th virus of the system in \eqref{eq:dt_SIR} converges to zero in exponential time.
\end{theorem}

The proof of Theorem~\ref{thm:GES} is inspired by the proof of \cite[Theorem 1]{gracy2020analysis}.

\textit{Proof:}
Consider an arbitrary virus $k \in [m]$. 
By Assumptions \ref{assume:two} and \ref{assume:three}, $M^k$, defined by Eq.~\eqref{eq:M}, is nonnegative for all $k\in [m]$, and from the condition we know that $\rho(M^k)< 1$. Therefore, according to Lemma~\ref{lemma:diagonal}, for each $k\in [m]$, there exists a positive definite diagonal matrix $P^k$ such that $(M^k)^\top P^k M^k - P^k$ is negative definite. 

Consider the candidate Lyapunov function: $V_1^k(t,x^k) = (x^k)^{\top} P^kx^k$. Since $P^k$ is diagonal and positive definite, $(x^k)^\top P^k x^k>0$, for all $x^k \neq \mathbf{0}$. 
Therefore, $V_1^k(t,x^k)>0$ for all $k\in[m], t \in \mathbb{Z}_{\geq 0}$, $x^k \neq \mathbf{0}$. 
Since $P^k$ is positive definite, 
\begin{equation}
  \lambda_{\text{min}}(P^k)I \leq P^k \leq  \lambda_{\text{max}}(P^k)I,    
\end{equation}
which implies that
\begin{equation}\label{eq:sigma2}
    \sigma_1^k\| x^k\|^2 \leq V_1^k(t,x^k) \leq  \sigma_2^k \| x^k\|^2,
\end{equation}
where 
$\sigma_1^k = \lambda_{\text{min}}(P^k)$ and $\sigma_2^k=  \lambda_{\text{max}}(P^k)$, with $\sigma_1^k, \sigma_2^k>0$.

Now we turn to computing $\Delta V_1^k(t,x^k)$. For $x^k\neq 0$ and for all $k \in[m]$, using \eqref{eq:SIR_dynamic} and \eqref{eq:M}-\eqref{eq:Mhat}, we have
\begin{align}
    &\Delta V_1^k(t,x^k) \nonumber \\ &
    = 
    (x^k)^\top \Tilde{M}^k [t]^\top P^k\Tilde{M}^k[t]x^k- (x^k)^\top P^k x^k \nonumber \\ 
    &= 
    (x^k)^\top [(M^k)^\top  P^kM^k- P^k]x^k \nonumber \\ & \ \  \ \ -2h (x^k)^\top (B^k)^\top (I-S[t])P^k M^k x^k \nonumber \\ 
    & \ \  \ \  +h^2 (x^k)^\top (B^k)^\top (I-S[t])P^k(I-S[t])B^k x^k.\label{eq:deltaV}
\end{align}
Note that the second and third terms of \eqref{eq:deltaV} can be reorganized as 
\begin{align}
    &(x^k)^\top [-2h (B^k)^\top (I-S[t])P^k M^k \nonumber \\&\ \   \ \   \ \ 
    +h^2 (B^k)^\top (I-S[t])P^k(I-S[t])B^k ] x^k \nonumber \\
    &= (x^k)^\top \{h (B^k)^\top (I-S[t])P^k\nonumber \\
    & \ \  \ \ [-2M^k+ h(I-S[t]) B^k]\}x^k\nonumber \\
    &= (x^k)^\top \{h (B^k)^\top (I-S[t])P^k\nonumber \\ \label{eq:second_third_term}
    & \ \  \ \ [-2(I -h\Gamma^k[t])-h(I+S[t])B^k ]\}x^k\leq 0, 
\end{align}
where the last equality follows from \eqref{eq:M}, and the inequality follows from Assumptions \ref{assume:two}-\ref{assume:three} and Lemma~\ref{lemma:one}. 
Thus, by plugging~\eqref{eq:second_third_term} into~\eqref{eq:deltaV}, we obtain
\begin{equation}\label{eq:lyap_thirdterm}
    \Delta V_1^k(t,x^k) \leq (x^k)^\top [(M^k)^\top  P^k M^k- P^k]x^k.
\end{equation}
Since $[(M^k)^\top  P^kM^k- P^k]$ is negative definite, 
we have, from Eq.~\eqref{eq:lyap_thirdterm},
\begin{equation}\label{eq:sigma3}
    \Delta V_1^k(t,x^k) \leq - \sigma_3^k \| x^k\|^2,
\end{equation}
where $\sigma_3^k =  \lambda_{\text{min}}[P^k-(M^k)^\top  P^kM^k]$, with $\sigma_3^k>0$.

Therefore, from~\eqref{eq:sigma2} and \eqref{eq:sigma3}, $V_1^k(t,x^k)$ is a Lyapunov function with exponential decay, and, hence, $x^k$ converges to zero at an exponential rate.~\qed




\begin{corollary}\label{coro:rate_exp}
Suppose that the conditions in Theorem~\ref{thm:GES} are fulfilled, and that 
$P^k$ is as defined in the proof of Theorem 1. Then the convergence of $x^k[t]$ 
has an exponential rate of at least $\sqrt{1-\frac{\sigma_3^k}{\sigma_2^k}}$, where $\sigma_2^k= \lambda_{\text{max}}(P^k)$, $\sigma_3^k = \lambda_{\text{min}}[P^k-(M^k)^\top  P^k M^k]$ for each $k \in[m]$.
\end{corollary}
\textit{Proof:}
From Lemma~\ref{lemma:rate_GES}, \eqref{eq:sigma2}, and \eqref{eq:sigma3}, the rate of convergence of virus $k$ is upper-bounded by $\sqrt{1-\frac{\sigma_3^k}{\sigma_2^k}}$.  We then need to show that the rate is well-defined, which is $\sqrt{1-\frac{\sigma_3^k}{\sigma_2^k}}  \in [0,1)$. Since $\sigma_2^k>0$ and $\sigma_3^k>0$, it will be sufficient to  show that $\sigma_2^k \geq \sigma_3^k$.

Since $P^k$ is positive definite and $(M^k)^\top P^k M^k $ is nonnegative definite, we have
\begin{equation}
    \sigma_3^k I \leq P^k - (M^k)^\top P^k M^k \leq P^k \leq \sigma_2^k I,
\end{equation}
from which $\sigma_2^k\geq \sigma_3^k$ and the rate of convergence $\sqrt{1-\frac{\sigma_3^k}{\sigma_2^k}}$ is well-defined.~\qed
\par In words, Theorem~\ref{thm:GES} says that if the state matrix corresponding to system~\eqref{eq:SIR_dynamic}, linearized around $s_i=1$ for $i \in [n]$, is Schur, then virus~k, irrespective of its initial infection levels in the population,  becomes extinct exponentially fast. 

We now present another sufficient condition, which depends on the spectral radius of $\tilde{M}^k[t]$,
for the infection level with respect to each of the viruses
to converge to zero at an exponential rate.
\begin{proposition}\label{prop:rho_tilde}
    Consider the system in~\eqref{eq:dt_SIR} under Assumptions \ref{assume:one}-\ref{assume:three}. If $\rho(\tilde{M}^k[t])<1$ for all $t \in \mathbb{Z}_{\geq 0}$, then the $k$-th virus of the system in \eqref{eq:dt_SIR} converges to zero in exponential time, with at least an exponential rate of $\min_{\forall t \in \mathbb{Z}_{\geq 0}}\Big[1-\rho(\tilde{M}^k[t])\Big]$.
\end{proposition}
\begin{proof}
    Define $\epsilon^k[t] := 1-\rho(\tilde{M}^k[t])$. Hence, we obtain that $\epsilon^k[t]>0$ for all $k\in [m]$ and $t \in \mathbb{Z}_{\geq 0}$. Since we have $x^k[t+1] = \tilde{M}^k[t] x^k[t]$ and $x^k[t] \geq 0$ for all $t \in \mathbb{Z}_{\geq 0}$, we can write that 
    \begin{equation}
    \frac{\lVert x^k[t+1] \rVert}{\lVert x^k[t] \rVert} = \frac{\lVert \tilde{M}^k[t]x^k[t] \rVert}{\lVert  x^k[t] \rVert} = |\lambda(\tilde{M}^k[t])|\leq \rho(\tilde{M}^k[t]),
\end{equation}
which results in
    $\lVert x^k[t+1] \rVert \leq(1-\epsilon^k[t]) \lVert x^k[t] \rVert$ for all $t$. Recall that, $\epsilon^k[t]>0$ for all $k\in[m]$ and $t \in \mathbb{Z}_{\geq 0}$. Then, we obtain that, for all $x_i^k[0]\in [0,1]^n$,
\begin{equation}\label{eq:converge_prop1}
    \lVert x^k[t] \rVert \leq [1- \min_{\forall t \in \mathbb{Z}_{\geq 0}}\epsilon^k[t]]^t \lVert x^k[0] \rVert  \leq e^{-t \min_{\forall t \in \mathbb{Z}_{\geq 0}}\epsilon^k[t]} \lVert x^k[0] \rVert, \nonumber 
\end{equation}
where the second inequality 
is due to
Bernoulli's inequality~\cite{carothers2000real}.
Hence, $x^k[t]$ converges to $\mathbf{0}$ with at least an exponential rate of $\min_{\forall t \in \mathbb{Z}_{\geq 0}}\epsilon^k[t]$.
\end{proof}
Note that $\rho(M^k)$ is the basic reproduction number of virus $k$,
the average number of infections produced by an infected individual in a population where everyone is susceptible,
and $\rho(\tilde{M}^k[t])$ is the effective reproduction number of virus $k$ over the network,
the average number of infected cases caused by 
an infected individual 
in a population made up of both susceptible and non-susceptible individuals.
%
It can be seen that Theorem~\ref{thm:GES} is a stronger global exponential stability condition than Proposition~\ref{prop:rho_tilde}, as $M^k$ is the linearized state transition matrix for virus $k$ which does not depend on the susceptible state. 
However, 
Proposition~\ref{prop:rho_tilde} is 
important 
in its own right, 
since it provides insights into the design of a feedback controller; 
see Section~\ref{sec:mitigation}.

\section{State Observation Model}\label{sec:observation}

In this section, we analyze the observation model introduced in Eq.~\eqref{eq:SIR_observationsub1} and Eq.~\eqref{eq:SIR_observationsub2}. In particular, we focus on identifying conditions guaranteeing strong local observability of our proposed system model around $s_i[t]=0$, and on estimating the system states with respect to each virus.
\par As a first step, we 
construct the observability matrix for the system 
by
writing Eq.~\eqref{eq:SIR_observationsub2} as: 
\begin{equation}\label{eq:observer_y_full}
    \mathbf{y}[t] = \mathbf{C} \mathbf{X}[t],
\end{equation}
where $\mathbf{y}[t] = \begin{bmatrix}
y_1[t] &
y_2[t] &
\cdots &
y_n[t]
\end{bmatrix}^\top \in \mathbb{R}^{n
}$;   the measurement matrix $\mathbf{C} \in \mathbb{R}^{n\times mn}$ is defined as: 

\begin{equation}
    \mathbf{C} = \begin{bmatrix}
C^1  & C^2& \cdots & C^m  
\end{bmatrix} \nonumber 
\end{equation}
with $C^k = \text{diag}([c_1^k, c_2^k, \cdots, c_n^k])$ for all $k \in [m]$; and
\begin{equation}
    \mathbf{X}[t] = \begin{bmatrix}
x^1[t]^\top  &
x^2 [t]^\top &
\cdots  &
x^m[t]^\top 
\end{bmatrix}^\top . \nonumber 
\end{equation}
Therefore, the measurement $\mathbf{y}[t] $ can be reorganized as:
\begin{align}
     \mathbf{y}[t]  
= C^1x^1[t]+C^2x^2[t]+\cdots+ C^mx^m[t].
\end{align}
The measurements, corresponding to each time step, over a time horizon $[t,t+m-1]$ can be concatenated in a vector as follows:
\begin{align}
    & \begin{bmatrix}
\mathbf{y}[t] \\
\mathbf{y}[t+1] \\
\mathbf{y}[t+2]\\
\vdots \\
\mathbf{y}[t+m-1]
\end{bmatrix}\nonumber \\ = &
\begin{bmatrix}
C^1x^1[t]+\cdots+ C^mx^m[t] \\
C^1x^1[t+1]+\cdots+ C^mx^m[t+1]  \\
C^1x^1[t+2]+\cdots+ C^mx^m[t+2] \\
\vdots \\
C^1x^1[t+m-1]+\cdots+ C^mx^m[t+m-1] 
\end{bmatrix}\nonumber 
\\ 
=& 
\begin{bmatrix}
C^1 \\
C^1\Tilde{M}^1[t]  \\
C^1\Tilde{M}^1[t]\Tilde{M}^1[t+1] \\
\vdots \\
C^1\Tilde{M}^1[t]\cdots \Tilde{M}^1[t+m-2]
\end{bmatrix}x^1[t] \nonumber  \\&+ \begin{bmatrix}
C^2  \\
C^2\Tilde{M}^2[t]  \\
C^2\Tilde{M}^2[t]\Tilde{M}^2[t+1] \\
\vdots \\
C^2\Tilde{M}^2[t]\cdots \Tilde{M}^2[t+m-2]
\end{bmatrix}x^2[t]+ \cdots \nonumber 
\\
&  + \begin{bmatrix}
C^m  \\
C^m\Tilde{M}^m[t]  \\
C^m\Tilde{M}^m[t]\Tilde{M}^m[t+1] \\
\vdots \\
C^m\Tilde{M}^m[t]\cdots \Tilde{M}^m[t+m-2]
\end{bmatrix}x^m[t]\nonumber  
\\  =& \mathcal{O}^1[t] x^1[t]+ \mathcal{O}^2[t] x^2[t]+ \cdots + \mathcal{O}^m[t] x^m[t]\nonumber \\
= &
\begin{bmatrix}
\mathcal{O}^1[t] &
\mathcal{O}^2[t] &
\cdots &
\mathcal{O}^m[t]
\end{bmatrix}\mathbf{X}[t]
,\label{eq:Ax=b}
\end{align}
where the matrix $\Tilde{M}^k[t]$ is as defined in Eq.~\eqref{eq:Mhat}, and
\begin{equation}
    \mathcal{O}^k[t] =  \begin{bmatrix}
C^k  \\
C^k\Tilde{M}^k[t]  \\
C^k\Tilde{M}^k[t]\Tilde{M}^k[t+1] \\
\vdots \\
C^k\Tilde{M}^m[t]\cdots \Tilde{M}^k[t+m-2]
\end{bmatrix}, \nonumber 
\end{equation}
with $\mathcal{O}^k[t]  \in \mathbb{R}^{mn\times n}$ for all $k \in [m]$.
We define the observability matrix of the system in Eq.~\eqref{eq:SIR_observation} as:
\begin{equation}
    \mathbb{O}[t] = \begin{bmatrix}
\mathcal{O}^1[t] &
\mathcal{O}^2[t] &
\cdots &
\mathcal{O}^m[t]
\end{bmatrix}, \label{eq:observability_matrix}
\end{equation}
where $\mathbb{O}[t] \in \mathbb{R}^{mn\times mn}$.



We are 
interested
in identifying a sufficient condition for strong local observability of our model when
the network consists only of infected and/or recovered individuals; in other words,
there are no susceptible individuals in any of the population nodes. Hence, we consider
the case when $s_i[t] =0, \forall i \in [n]$. 
Then the observability matrix in Eq.~\eqref{eq:observability_matrix} becomes 
\begin{equation}
    \mathbb{O}_0[t] = \begin{bmatrix}
\mathcal{O}_0^1[t] &
\mathcal{O}_0^2[t] &
\cdots &
\mathcal{O}_0^m[t]
\end{bmatrix}, \label{eq:observability_matrix_zero}
\end{equation}
where
\begin{equation}
    \mathcal{O}_0^k[t] =  \begin{bmatrix}
C^k  \\
C^k(I-h\Gamma^k)  \\
C^k(I-h\Gamma^k)^2 \\
\vdots \\
C^k (I-h\Gamma^k)^{m-1} 
\end{bmatrix},\label{eq:observability_matrix_block_zero}
\end{equation}
for all $k \in [m]$.
\par We have the following result.
\begin{theorem}\label{thm:Observable_zero}
Consider the system in Eq.~\eqref{eq:SIR_observation} 
under Assumptions~\ref{assume:one}-\ref{assume:four}. Suppose that $y_i(t)$, for all $i \in [n]$ and $t \in [t, t+m-1]$ is known. If, for each $i \in [n], \gamma_i^k \neq \gamma_i^{k_0}$ for all $k, k_0 \in [m]$ and $k \neq k_0$, then the system in Eq.~\eqref{eq:SIR_observation}  is strongly locally observable at $s_i(t)=0$ for all $i \in [n]$.
\end{theorem}
\begin{proof}
From the assumptions $h\sum_{k=1}^m \gamma^k_i\leq 1$, $h>0$, and $\gamma_i^k>0$ for all $i\in [n], k \in [m]$, 
we obtain that 
$1-h\gamma_i^k \in (0,1)$ 
for all $i\in [n], k \in [m]$. In addition, 
by Assumption~\ref{assume:four}, $c_i^k \in (0,1]$
for all $i\in [n], k \in [m]$, we can conclude that the entries of Eq.~\eqref{eq:observability_matrix_zero}: $c_i^k(1-h\gamma_i^k)\in (0,1)$ for all $i\in [n], k \in [m]$.

We let $\textbf{0}^{n} := \begin{bmatrix}
0 & 0& \cdots & 0
\end{bmatrix}^{1\times n}$ and $\textbf{0}^{0} := \emptyset$.
Consider Eq.~\eqref{eq:observability_matrix_block_zero} and recall that every 
block matrix 
is diagonal; hence, Eq.~\eqref{eq:observability_matrix_zero} is the concatenation of a set of block diagonal matrices. 
For all $i\in[n]$, the $i$-th row 
of the observability matrix~\eqref{eq:observability_matrix_zero} can be written as: 
\begin{equation}
    \begin{bmatrix}
\textbf{0}^{i-1}& c_i^1 &\textbf{0}^{n-i} &\textbf{0}^{i-1} &c_i^2& \textbf{0}^{n-i}& \cdots & \textbf{0}^{i-1} &c_i^m  &\textbf{0}^{n-i}
\end{bmatrix}\nonumber 
\end{equation}
which is linearly independent with the $(i+ln)$-th row of~\eqref{eq:observability_matrix_zero} for all $l \in [m-1]$:
\small
\begin{equation}
    \begin{bmatrix}
 \textbf{0}^{i-1}& c_i^1(1-h\gamma_i^1)^{l} &\textbf{0}^{n-i}&  
 \cdots & \textbf{0}^{i-1}  & c_i^m(1-h\gamma_i^m)^{l} &\textbf{0}^{n-i}
\end{bmatrix} \nonumber 
\end{equation}

\normalsize

\noindent
under our assumption that, for each $i \in [n]$, 
$\gamma_i^k \neq \gamma_{i}^{k_0}$ for all $k,k_0 \in[m]$ and $k \neq k_0$.
Thus, the observability matrix in Eq.~\eqref{eq:observability_matrix_zero} has full row rank. Since the observability matrix is a square matrix, we conclude that the observability matrix, $\mathbb{O}_0[t]$, 
is full rank.
Therefore, the mapping in Eq.~\eqref{eq:Ax=b} when $s_i[t]=0, 
\forall i \in [n]$ is injective. Note that $mn$ is the dimension of $\mathbf{X}[t]$. Hence, by Lemma \ref{lemma:locally_observable},
the competing virus model in~\eqref{eq:SIR_observation} 
is strongly locally observable at $s_i[t]=0, 
\forall i \in [n]$.
\end{proof}

Notice that whenever we add another virus to 
our model~\eqref{eq:SIR_observation},
we increase the dimension of~\eqref{eq:observability_matrix_zero} from $mn\times mn$ to $(m+1)n \times (m+1)n$, 
and, due to the same reasons as in the proof of Theorem~\ref{thm:Observable_zero}, the rank of the observability matrix will 
be $(m+1)n$. 

\begin{remark}
The assumption in Theorem~\ref{thm:Observable_zero}, namely that, for each $i \in [n]$, $\gamma_i^k$
is a distinct value across every $k \in[m]$, 
can be interpreted as 
each node's recovery rates with respect to every virus are distinct.
This assumption is 
reasonable 
as the recovery rate 
represents the inverse of the average duration of an infected individual to be sick, and the average amount of time for an individual to recover from different types/strains of viruses varies drastically\cite{whitley2001herpes}.
\end{remark}

Theorem~\ref{thm:Observable_zero} provides a sufficient condition for strong local observability when the fraction of susceptible population in each node is zero. Note that this condition 
identifies a scenario that admits the design of an observer for estimating the system states, but it does not say \textit{how} the states of the system can be estimated. Hence, in the sequel, we focus on the estimation of the system states.\\
The 
dynamics of the estimated states are 
\begin{align}
    \hat{x}^k_i[t+1] &= \hat{x}^k_i[t]  +h\bigg\{ \hat{s}_i[t]\sum_{j=1}^n \beta^k_{ij} \hat{x}^k_j[t]-\gamma^k_i \hat{x}^k_i[t]\bigg\}, \label{eq:infection_est} \\
    \hat{y}_i[t+1] & = \sum_{k=1}^m c_i^k \hat{x}_i^k[t+1],  \label{eq:observer}
\end{align}
where
\begin{align}\label{eq:hat:s}
    \hat{s}_i[t] &= 1-\sum_{k=1}^m\hat{x}^k_i[t]-\hat{r}_i[t] 
\end{align}
and
the recovered level is estimated by:
\begin{equation}\label{eq:hat:r}
    \hat{r}_i[t] = h\sum_{q=0}^t \sum_{k=1}^m \gamma_i^k \hat{x}_i^k[q]
\end{equation}
at each time step, recursively. 
Notice that in Eq.~\eqref{eq:infection_est}, in order to acquire the estimated infection level at node $i$, we 
need some knowledge of the
infection levels from all the neighbors
of node $i$, namely $\hat{x}^k_j[t]$ for all $\beta_{ij}^k \neq 0$. 
Hence, we have the following definition and assumption.
\begin{definition}\label{definition:report}
For node $j$, which is a neighbor node of node~$i$, namely $\beta_{ij}^k > 0$, we define the estimated infection level at node $j$ acquired by node $i$ at time $t$ as $\hat{x}^k_j[t-\mathcal{T}_j]+\mathcal{E}_j$, where $\mathcal{T}_j \in \mathbb{Z}_{\geq 0}$ is the time delay between nodes $i$ and $j$, and $\mathcal{E}_j\in \mathbb{R}$ is the reporting error at node $j$.
\end{definition}
\begin{assumption}\label{assume:cooperative}
For the estimation algorithm, 
the nodes 
share their estimated infection levels with no time delay or error, namely $\mathcal{T}_j=0, \mathcal{E}_j=0$ for all $j \in [n]$ in Definition~\ref{definition:report}.
\end{assumption}

Through Assumption~\ref{assume:cooperative}, we have that every node in the network is completely cooperative and honest to its neighboring nodes. Hence, under
Assumption~\ref{assume:cooperative}, our proposed distributed Luenberger observer is:
\begin{align}\label{eq:SIR_luenberger_observation}
    \hat{x}^k_i[t+1] =& \hat{x}^k_i[t] +h\bigg\{\hat{s}_i[t]\sum_{j=1}^n \beta^k_{ij} \hat{x}^k_j[t]-\gamma^k_i \hat{x}^k_i[t]\bigg\} \nonumber  \\& \ \   \ \ +L_i^k(y_i[t]-\hat{y}_i[t]),
\end{align}
where $L_i^k$ is the observer gain which, given a node $i$, can be chosen for each $k \in [m]$. 
We can write $e_i^k[t]=x_i^k[t]-\hat{x}_i^k[t]$ as the error of the observer. Hence, the dynamics of the estimation error are written as:
\begin{align}
    e_i^k[t+1] &=x_i^k[t+1]-\hat{x}_i^k[t+1]  
    \nonumber  \\
     &= (1-h\gamma_i^k)e_i^k[t]+h\sum_{j=1}^n \beta^k_{ij} e^k_j[t]
     -L_i^k\sum_{k=1}^m c_i^k e_i^k[t] \nonumber \\  &-h \bigg(\sum_{k=1}^m x_i^k[t] \sum_{j=1}^n \beta_{ij}^k x_j^k[t]- \sum_{k=1}^m \hat{x}_i^k[t] \sum_{j=1}^n \beta_{ij}^k \hat{x}_j^k[t]\bigg)\nonumber \\
     & - h\bigg(  \sum_{q}^t \sum_{k=1}^m \gamma_i^k x_i^k[q] \sum_{j=1}^n \beta_{ij}^k x_j^k[t] \nonumber  \\ &
     - \sum_{q}^t \sum_{k=1}^m \gamma_i^k \hat{x}_i^k[q] \sum_{j=1}^n \beta_{ij}^k \hat{x}_j^k[t] \bigg). \label{eq:error_eq}
\end{align}
We then rewrite the error dynamics~\eqref{eq:error_eq} as:
\begin{align}\label{eq:estimation_vector}
    e^k[t+1] =(M^k-L^k C^k) e^k[t] + w^k[t]
\end{align}
where recall that
\begin{equation}
    M^k = I-h\Gamma^k +hB^k \nonumber 
\end{equation}
and
\begin{align}\label{eq:nonlinear}
    &w^k[t] = -L^k \sum_{p\neq k}^m c^p e^p[t] \nonumber \\ &
    +h\bigg\{ \Big(\sum_{k=1}^m x^k[t]\Big)B^k x^k[t]-\Big(\sum_{k=1}^m\hat{x}^k[t]\Big)B^k \hat{x}^k[t] \bigg\} \nonumber \\
    &+h\bigg\{ \sum_{q}^t \sum_{k=1}^m \Gamma^k x^k[q]  B^k x^k[t] - \sum_{q}^t \sum_{k=1}^m \gamma_i^k \hat{x}_i^k[q]  B^k \hat{x}^k[t] \bigg\}
\end{align}
with $L^k = \text{diag}(L_i^k)$. 

Inspired by~\cite{niazi2022observer}, we 
aim to show that the estimation error of our Luenberger observer~\eqref{eq:SIR_luenberger_observation} converges to zero asymptotically. We make the following assumption.

\begin{assumption}\label{assume:lipschitz}
There exists a constant $l^k$ for each virus $k\in [m]$ such that:
\begin{equation}
    ||w^k[t]|| \leq l^k ||x^k[t]-\hat{x}^k[t]||
    \label{eq:lipschitz}
\end{equation}
for all 
$t \in \mathbb{Z}_{\geq 0}$.
\end{assumption}

\begin{corollary}\label{coro:lipschitz_const}
For each $k\in [m]$, 
inequality~\eqref{eq:lipschitz} 
can be rewritten as:
\begin{equation}
    l^k \geq \frac{||w^k[t]|| }{||x^k[t]-\hat{x}^k[t]||}
\end{equation}
for all $t \in \mathbb{Z}_{\geq 0}$.
\end{corollary}
%
Note that,
in many cases, we are able to tune the observer gain $L^k$ 
to satisfy the inequality~\eqref{eq:lipschitz} in Assumption~\ref{assume:lipschitz}.
We explore the feasibility of Assumption~\ref{assume:lipschitz} 
via simulations in Section~\ref{sec:simulation}.
We now identify a sufficient condition with respect to the observer gain, which guarantees that the estimation errors converge to zero asymptotically.



\begin{theorem}\label{thm:Error_GAS}
Under Assumption~\ref{assume:lipschitz}
, the estimation error of the Luenberger observer~\eqref{eq:SIR_luenberger_observation} for virus $k$ converges to zero asymptotically if there exist a symmetric matrix $Q^k\succ 0$ 
and $\tau^k\in (0,1]$ such that the following 
inequalities are
satisfied:
\begin{align}
     &(M^k-L^kC^k)^\top[Q^k-(Q^k)^\top (Q^k-\tau^k I)^{-1}Q^k](M^k-L^kC^k) \nonumber  \\&-Q^k+\tau^k(l^k)^2I\prec0 \label{eq:lmi_assumption} 
     \end{align}
     and
     \begin{equation}
 Q^k-\tau^k I \prec0.\label{eq:lmi_assumption_II}
\end{equation}
where $L^k$ is 
the 
observer gain for virus $k$. 
\end{theorem}

\begin{proof}
From Assumption~\ref{assume:lipschitz}, we can write that
\begin{equation}
    (w^k[t])^\top w^k[t] \leq (l^k)^2  (e^k[t])^\top e^k[t] \nonumber
\end{equation}
which can be rewritten as 
\begin{equation}
        \begin{bmatrix}
    (e^k[t])^\top & (w^k[t])^\top
    \end{bmatrix}
    \Phi^k
    \begin{bmatrix}
    e^k[t] \\ w^k[t]
    \end{bmatrix}\leq 0, \label{eq:Phi_zero}
\end{equation}
where \begin{equation}
    \Phi^k = \begin{bmatrix}
    -(l^k)^2 I & 0\\ 0 & I
    \end{bmatrix}. \nonumber
\end{equation}
By utilizing the Schur complement~\cite{zhang2006schur}
, 
and defining $A^k = (M^k-L^kC^k)^\top Q^k(M^k-L^kC^k)$, where $Q^k$ is a symmetric positive definite matrix by assumption,
we can reorganize Eqs.~\eqref{eq:lmi_assumption} and~\eqref{eq:lmi_assumption_II} as 
\begin{align}
     \begin{bmatrix}
A^k-Q^k+\tau^k (l^k)^2 I & (M^k-L^kC^k)^\top Q^k \\
Q^k(M^k-L^kC^k) & Q^k -\tau^k I
\end{bmatrix}
\prec 0,    \label{eq:lmi_assumption_2x2}
\end{align}
which yields the following:
\begin{equation}
\Omega^k - \tau^k \Phi^k
   \prec0, \label{eq:Omega_zero}
   \end{equation}
   where 
\begin{align}
    &\Omega^k =  \begin{bmatrix}
    A^k-Q^k & (M^k-L^kC^k)^\top Q^k \\
    Q^k(M^k-L^kC^k) & Q^k
    \end{bmatrix}. \nonumber 
\end{align}

We now consider the candidate Lyapunov function $V_2^k(e^k,[t]) = (e^k[t])^\top Q^k (e^k[t])$, where $Q^k$ is the positive definite matrix in~\eqref{eq:lmi_assumption_2x2}. 
We can write that
\begin{equation}
    V_2^k(e^k,[t])>0, \text{ for all } e^k[t] \neq 0 \nonumber 
\end{equation}
and
\begin{align}
    \Delta &V_2^k \nonumber  \\ =& V_2^k(e^k[t+1],t)- V_2^k(e^k[t],t) \nonumber  \\
    =& \Big[(M^k-L^kC^k) e^k[t]+w^k[t]\Big]^\top  Q^k \nonumber  \\ & \Big[(M^k-L^kC^k) e^k[t]+w^k[t]\Big] 
    - (e^k[t])^\top Q^k e^k[t] \nonumber \\
     =&(e^k[t])^\top 
     A^k
     e^k[t] \nonumber  \\
     &+(w^k[t])^\top Q^k (M^k-L^kC^k) e^k[t] \nonumber  \\ &  +(e^k[t])^\top (M^k-L^kC^k)^\top Q^kw^k[t]+ (w^k[t])^\top  Q^k w^k[t] \nonumber \\ &-(e^k[t])^\top Q^k e^k[t] \nonumber \\
    =&(e^k[t])^\top  
    \Big(A^k
    -Q^k\Big) e^k[t] \nonumber \\
    &  +(w^k[t])^\top Q^k (M^k-L^kC^k) e^k[t] \nonumber  \\
    &+(e^k[t])^\top (M^k-L^kC^k)^\top Q^k w^k[t] \nonumber  \\ &+(w^k[t])^\top  Q^k w^k[t] \nonumber 
    \\ =& \begin{bmatrix}
    (e^k[t])^\top & (w^k[t])^\top
    \end{bmatrix}\Omega^k \begin{bmatrix}
    e^k[t] \\ w^k[t]
    \end{bmatrix} \nonumber 
    \\ <& \begin{bmatrix}
    (e^k[t])^\top & (w^k[t])^\top
    \end{bmatrix} \tau^k \Phi^k  \begin{bmatrix}
    e^k[t] \\ w^k[t]
    \end{bmatrix} \label{eq:Omega_Phi}
    \\ \leq  & 0, \label{eq:Lyap_zero}
\end{align}
where 
inequality~\eqref{eq:Omega_Phi} follows from 
~\eqref{eq:Omega_zero},
and 
inequality~\eqref{eq:Lyap_zero} 
follows from
inequality~\eqref{eq:Phi_zero}.
Therefore, by Lyapunov's direct method~\cite{vidyasagar2002nonlinear}, the estimation error of 
the Luenberger observer~\eqref{eq:SIR_luenberger_observation} for virus $k$ is globally asymptotically stable.
\end{proof}

\begin{remark}
Note that we need the pair $(M^k,C^k)$ to be detectable for inequalities~\eqref{eq:lmi_assumption} and~\eqref{eq:lmi_assumption_II} to hold~\cite{bara2005observer, niazi2022observer}. Under Assumption~\ref{assume:four}, the pair $(M^k,C^k)$ is observable, and, thus, detectable.
\end{remark}


\begin{corollary}\label{coro:lmi}
The inequalities~\eqref{eq:lmi_assumption} and~\eqref{eq:lmi_assumption_II} combined together in Theorem~\ref{thm:Error_GAS} 
are equivalent to
the following LMI:
\begin{align}
\tiny \begin{bmatrix} 
-Q^k+\tau^k  (l^k)^2I & (M^k)^\top Q^k -(C^k)^\top R^k & (M^k)^\top Q^k -(C^k)^\top R^k \\
[(M^k)^\top Q^k -(C^k)^\top R^k]^\top & Q^k - \tau^k I & \mathbf{0} \\
[(M^k)^\top Q^k -(C^k)^\top R^k]^\top & \mathbf{0} & -Q^k
\end{bmatrix}& \nonumber \\ \prec0&
\label{eq:lmi}
\end{align}
where $(R^k)^\top = Q^k L^k $.
\end{corollary}
By applying Schur complement, LMI~\eqref{eq:lmi} can be shown to be equivalent to Eq.~\eqref{eq:lmi_assumption_2x2} and $-Q^k \prec 0$. Note that LMI~\eqref{eq:lmi} can be
solved through the $cvx$ solver in MATLAB for simulations.

\section{Distributed Feedback Control}\label{sec:mitigation}

In this section, we present a distributed feedback mitigation strategy for ensuring that
all viruses are eradicated. We establish that virus $k$ can be eradicated in exponential time by boosting the healing rate associated with virus $k$.
Applying 
such eradication strategy for
all $m$ viruses, the system converges
to a healthy state.

When battling against the spread of an epidemic, boosting the healing rate with respect to the virus is a common approach~\cite{zhang2021estimation, gracy2020analysis}. 
Boosting the healing rates could be implemented by means of
providing effective medication, medical supplies, and/or healthcare workers to each subpopulation.
\par The key tool behind devising the aforementioned mitigation strategy is Proposition~\ref{prop:rho_tilde}, which
says that
if the spectral radius of the state transition matrix of virus  $1$ is less than one, i.e.,  $\rho(\tilde{M}^k[t])<1$,  then the infection level of the $k$-th virus converges to zero within at least exponential time. 
Accordingly, we formally state our 
distributed feedback control strategy as follows: 
\begin{equation}\label{eq:control_scheme_feedback}
    \widetilde{\gamma}_i^k[t] = \gamma_i^k -u^k_i[t],  \;\ i \in [n],
\end{equation}
where $u^k_i[t]$ is a state feedback controller, with
\begin{align}
    u^k_i[t] &= -s_i[t]\sum_{j=1}^n \beta_{ij}^k \nonumber \\
    & = -\bigg(1-\sum_{k=1}^m x_i^k[t]-r_i[t]\bigg)\sum_{j=1}^n \beta_{ij}^k.\label{eq:control_u}
\end{align}
We have the following result.
\begin{theorem}\label{thm:control_feedback}
Consider the system in~\eqref{eq:dt_SIR}, under Assumptions~\ref{assume:two} and~\ref{assume:three}, and assume further that $h\widetilde{\gamma}_i^k[t] <1$, 
    $\forall i \in [n], t \in \mathbb{Z}_{\geq 0}$.
Then the 
feedback controller in Eq.~\eqref{eq:control_scheme_feedback} 
guarantees that virus $k$ is eradicated with at least an exponential rate of $h\min_{i \in [n]} \{\gamma_i^k\}$.
\end{theorem}
\begin{proof}
By substituting Eq.~\eqref{eq:control_scheme_feedback} and Eq.~\eqref{eq:control_u} into \eqref{eq:dt_SIRsub2}, we obtain
\begin{align}\label{eq:infected_control}
    x_i^k[t+1] =& x_i^k[t] + \nonumber \\ & h\Bigg\{s_i[t]\sum_{j=1}^n \beta_{ij}^k x_j^k[t] - \Big[ s_i[t]\sum_{j=1}^n \beta_{ij}^k +\gamma_i^k \Big]x_i^k[t]\Bigg\}. 
\end{align}
The state transition matrix of \eqref{eq:infected_control} can be written as 
\begin{equation}
    \widetilde{M}^k[t] = I+h\Big[S[t]B^k -\big(S[t]\text{diag}(B^k \mathbf{1}_{n\times 1}) + \text{diag}(\gamma_i^k)\big)\Big]. \nonumber 
\end{equation}
The entries of the $i$-th row of $\widetilde{M}^k[t]$, therefore, are
\begin{equation}
    \widetilde{m}_{ii}^k[t] = 1+h\Bigg[s_i[t]\beta_{ii}^k-s_i[t]\sum\limits_{j=1}^n \beta_{ij}^k-\gamma_i^k\Bigg], \nonumber 
\end{equation}
\begin{equation}
    \widetilde{m}_{ij}^k[t] = h s_i[t]\beta_{ij}^k, \nonumber 
\end{equation}
which satisfies the following inequality
\begin{equation}
    \widetilde{m}_{ii}^k[t]+ \sum\limits_{j\neq i}^n\widetilde{m}_{ij}^k [t] \leq 1- 
    \gamma_i^k, 
    \forall i \in n. \nonumber
\end{equation}
Therefore, by Gershgorin circle theorem, the spectral radius of $\widetilde{M}^k [t] $ is upper bounded by $1- h\min_{i \in [n]} \{\gamma_i^k\}$, that is, 
\begin{equation}
    \rho(\widetilde{M}^k [t]) \leq 1- h\min_{i \in [n]} \{\gamma_i^k\}.  \nonumber 
\end{equation}
Since we have $x^k[t+1] = \widetilde{M}^k[t] x^k[t]$ and $x^k[t] \geq 0$ for all $t \in \mathbb{Z}_{\geq 0}$, we can write that 
\begin{equation}
    \frac{\lVert x^k[t+1] \rVert}{\lVert x^k[t] \rVert} = \frac{\lVert \widetilde{M}^k[t]x^k[t] \rVert}{\lVert  x^k[t] \rVert} = |\lambda(\widetilde{M}^k[t])|\leq \rho(\widetilde{M}^k[t]),
\end{equation}
which results in
$\lVert x^k[t+1] \rVert \leq[1- h\min_{i \in [n]} \{\gamma_i^k\}] \lVert x^k[t] \rVert$ for all $t$. Since, $\gamma_i^k>0$ for all $i \in [n], k\in[m]$, from Assumption~\ref{assume:two}, we obtain that, for all $x_i^k[0]\in [0,1]^n$,
\begin{equation}\label{eq:converge}
    \lVert x^k[t] \rVert \leq [1- h\min_{i \in [n]} \{\gamma_i^k\}]^t \lVert x^k[0] \rVert  \leq e^{-t h\min_{i \in [n]} \{\gamma_i^k\}} \lVert x^k[0] \rVert,
\end{equation}
where the second inequality holds by
Bernoulli's inequality~\cite{carothers2000real}.
Hence, $x^k[t]$ converges to $\mathbf{0}$ with at least an exponential rate of $h\min_{i \in [n]} \{\gamma_i^k\}$, according to Proposition~\ref{prop:rho_tilde}.
\end{proof}

\begin{corollary}
If we apply the mitigation strategy in Theorem~\ref{thm:control_feedback} to all viruses $k \in[m]$, then the system in~\eqref{eq:dt_SIR} converges to the healthy state at an exponential rate.
\end{corollary}

\begin{remark}
The control strategy proposed in Theorem~\ref{thm:control_feedback} can be interpreted as follows: if the healing rate with respect to virus $k$ of each subpopulation is appropriately increased according to the susceptible level, 
for example by distributing effective medication, medical supplies, and healthcare workers to each subpopulation, 
then the epidemic will be eradicated at an exponential rate.
This theorem provides decision-makers
insight into, given sufficient resources, how to allocate
medical supplies and healthcare workers
to different subpopulations
so that the epidemic can be eradicated quickly. Moreover, notice that the convergence rate of virus $k$ depends on the minimum healing rate at each node of the network corresponding to the $k$-th virus, which encourages the decision-makers to elevate the 
lowest
healthcare level of the subpopulation within the community.
\end{remark}

Our feedback controller~\eqref{eq:control_scheme_feedback} is an improvement with respect to similar \textit{open-loop} control schemes for networked SIS models in \cite{ye2021applications,gracy2020analysis,liu2019analysis}, 
since it does not require 
full information of the system states. 
More specifically, for the feedback gain design,
because the recovered state can be calculated
using  $r_i[t]=h\sum_{q=0}^t \sum_{k=1}^m \gamma_i^k x_i^k[q]$,
we only need 
knowledge of the susceptible level $s_i[t]$ or 
 knowledge of the infected 
level, $x_i^k[t]$ for all $k\in [m], i\in [n]$.
Furthermore, compared to the control strategies proposed in~\cite{liu2019analysis, gracy2020analysis, ye2021applications} whose boosted healing rates maintain constant values, our distributed feedback controller, since it updates the healing rates in response to the infection level in a subpopulation, is capable of allocating the medical resource more efficiently.
In the next section, among other results, we will explore the performance of our distributed feedback controller by relying only on the estimates of the system states in addition to the actual system states.

\section{Simulations}\label{sec:simulation}

In this section, we consider the special case of 
two competing variants of the SARS-CoV-2 virus,
(i.e. 
$m=2$): 
Delta and Omicron
spreading over the network depicted in Figure~\ref{fig:Graph_Europe}~\cite{zipfel2021missing}.
We choose the two variants of the SARS-CoV-2 virus 
because they cause patients to display similar symptoms such as fever, coughing, 
and 
headache
\cite{antonelli2022risk, rader2022use}.
We consider a network of $5$ nodes, where
each node represents a country in Europe: 
France, Italy, Switzerland, Austria, and Germany; there exists an edge between two nodes if the countries that 
the nodes represent 
share a border with each other. 
The system parameters corresponding to the two variants of SARS-CoV-2 virus 
are listed in Table~\ref{table:parameters_COVID} and Table~\ref{table:parameters_Flu}, respectively.
This section includes no real data; however, the viral spreading parameters are inspired by the behavior of the viruses~\cite{shrestha2022evolution},
that is, the model parameters are chosen so that the 
Omicron variant
is more contagious than 
Delta.
We also acknowledge that these two variants
are not necessarily competitive; there are cases where people have been infected with both variants.
However, the samples of co-infection of both variants are very uncommon~\cite{bolze2022evidence}, thus can be disregarded in the population and the time scales being considered in these simulations.

\begin{table}[h!]
\centering
\begin{tabular}{|c | c c c c c|} 
 \hline  
 $\beta_{ij}^1$  
 & FR & IT & CH & AT & DE  \\ [0.5ex] 
 \hline
 FR   & 0.08& 0.15& 0.24& 0& 0.06 \\
 IT & 0.15& 0.12& 0.13& 0.11& 0  \\
 CH       & 0.24& 0.13& 0.25& 0.05& 0.04  \\
 AT        & 0& 0.09& 0.05& 0.11& 0.15   \\
 DE       &  0.06& 0& 0.04& 0.14& 0.09  \\ \hline\hline
 $\gamma_i^1$       &0.15& 0.23& 0.17& 0.25& 0.2  \\
 $x^1[0]$       & 0.005       & 0.01          & 0.0075   & 0.0025  & 0.0075 
  \\$c^1$       & 0.4       &  0.4          &  0.4   &  0.4  & 0.4   \\ 
 \hline
\end{tabular}
\caption{For the network in Figure~\ref{fig:Graph_Europe}, disease parameters corresponding to the Omicron variant.}
\label{table:parameters_COVID}
\end{table}

\begin{table}[h!]
\centering
\begin{tabular}{|c | c c c c c|} 
 \hline
 $\beta_{ij}^2$   & FR & IT & CH & AT & DE  \\ [0.5ex] 
 \hline
 FR   & 0.02& 0.05& 0.04& 0& 0.01 \\
 IT & 0.05& 0.06& 0.07& 0.02& 0  \\
 CH       & 0.04& 0.07& 0.04& 0.03& 0.05  \\
 AT        & 0& 0.03& 0.04& 0.09& 0.07   \\
 DE       &  0.01& 0& 0.05& 0.07& 0.06 \\ \hline\hline
 $\gamma_i^2$       &0.095& 0.12& 0.1& 0.15& 0.13  \\
 $x^2[0]$       & 0.001       & 0.002          & 0.0035   & 0.002  & 0.001   \\$c^2$       & 0.3       &  0.3          &  0.3   &  0.3  & 0.3\\ 
 \hline
\end{tabular}
\caption{For the network in Figure~\ref{fig:Graph_Europe}, disease parameters corresponding to the Delta variant.}
\label{table:parameters_Flu}
\end{table}

\begin{figure}
\centering
\includegraphics[width=.3\textwidth]{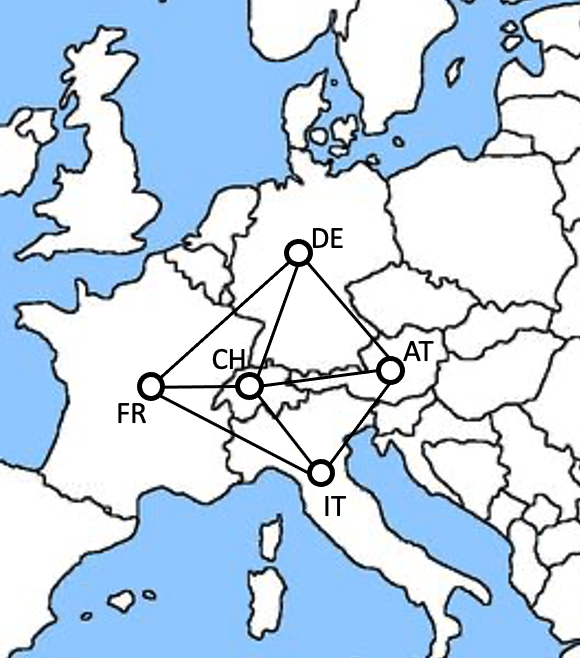}
\caption{Contact graph for the spread of the
Omicron and Delta variants of
SARS-CoV-2 virus.}
\label{fig:Graph_Europe}
\end{figure}


\begin{figure}
\centering
\begin{overpic}[width = 0.48\columnwidth]{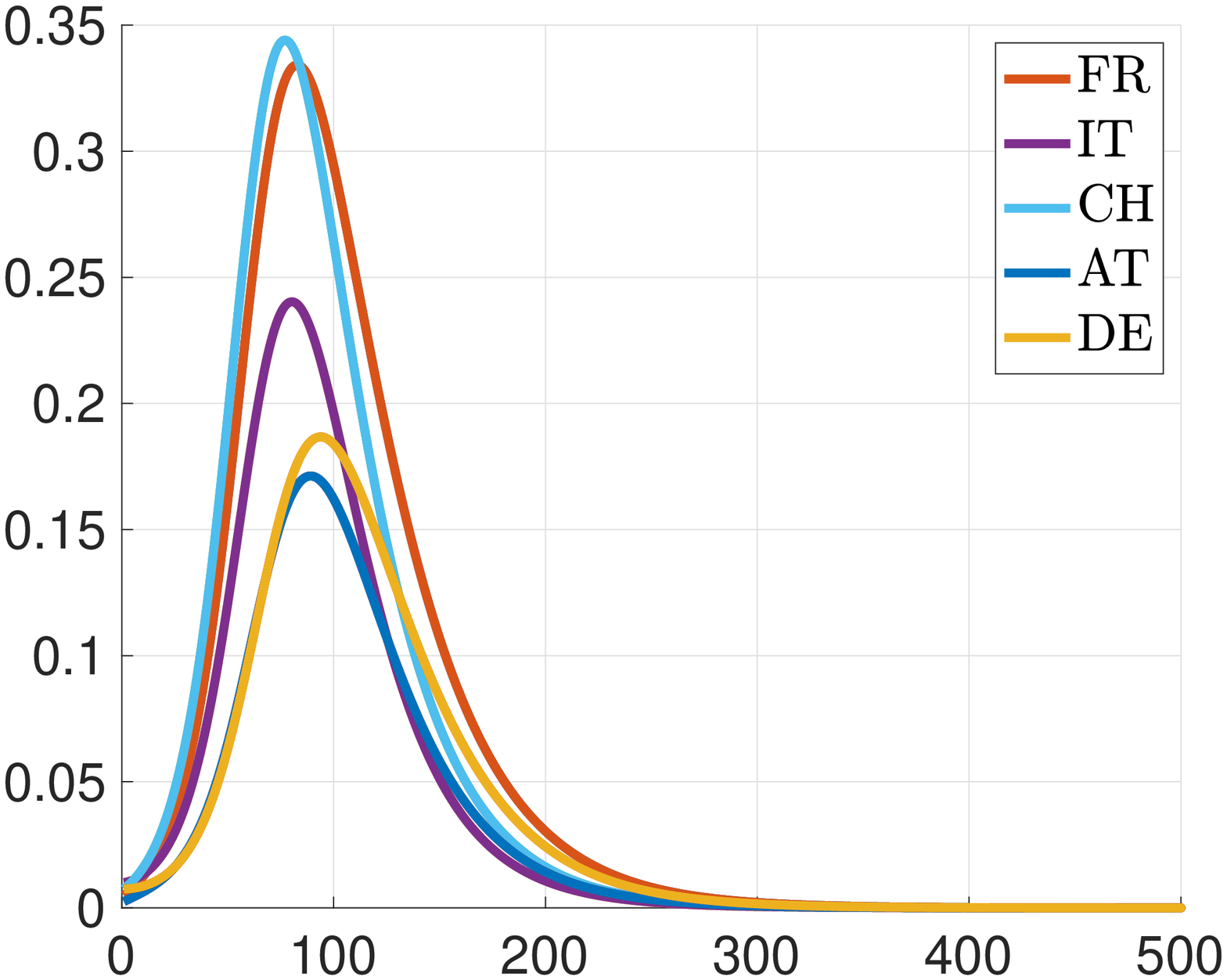}
     \put(-2,35){{\parbox{0.75\linewidth}\footnotesize \rotatebox{90}{\footnotesize$x^1$}
     }}
     \put(50,-1){\footnotesize{\parbox{0.75\linewidth}\footnotesize $t$
     }}\normalsize
   \end{overpic}
\begin{overpic}[width =  0.48\columnwidth]{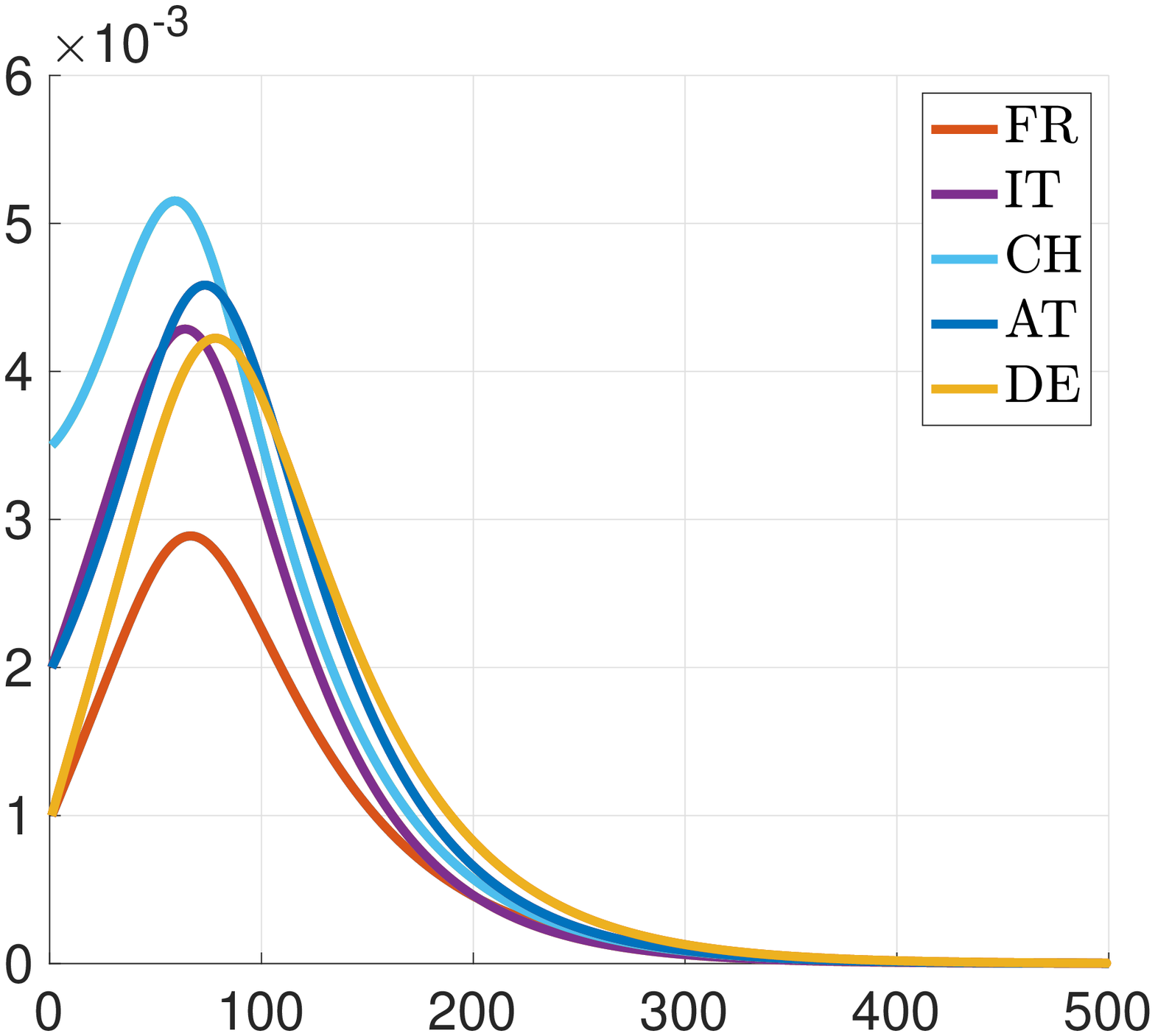}
     \put(-0,35){{\parbox{0.75\linewidth}\footnotesize \rotatebox{90}{\footnotesize$x^2 $}
     }}\normalsize
     \put(50,-1){\footnotesize 
    $t$
     }
   \end{overpic}
\caption{Evolution of infection level of the
Omicron variant
in each country (left); Evolution of infection level of the Delta variant in each country (right).}
\label{fig:Evolution}
\end{figure}

\begin{figure}
\centering
\begin{overpic}[width = 0.48\columnwidth]{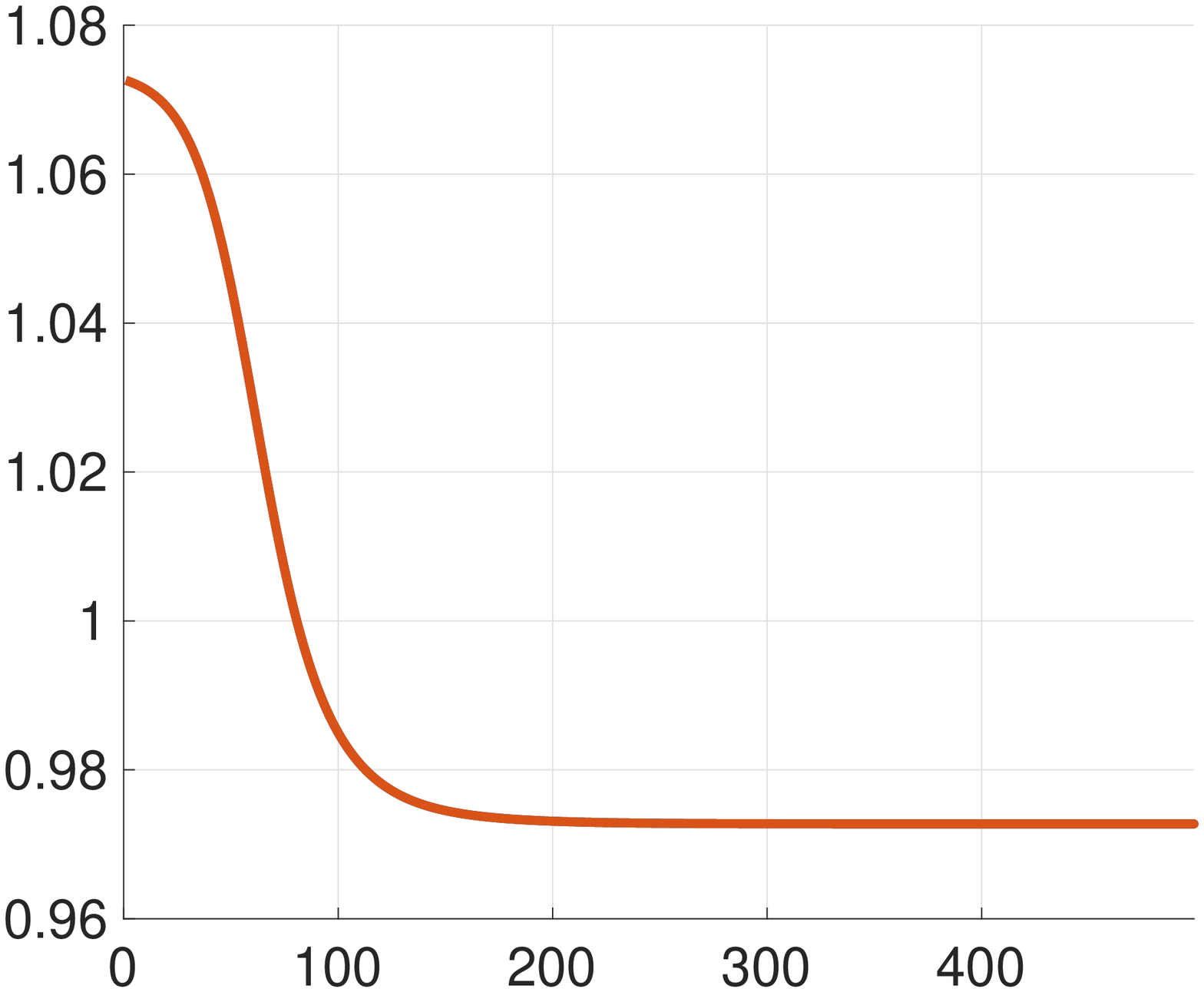}
     \put(-5,25){{\parbox{0.75\linewidth}\footnotesize \rotatebox{90}{\footnotesize$\rho(\Tilde{M}^1[t])$}
     }}
     \put(50,0){\footnotesize{\parbox{0.75\linewidth}\footnotesize $t$
     }}\normalsize
  \end{overpic}
\begin{overpic}[width =  0.48\columnwidth]{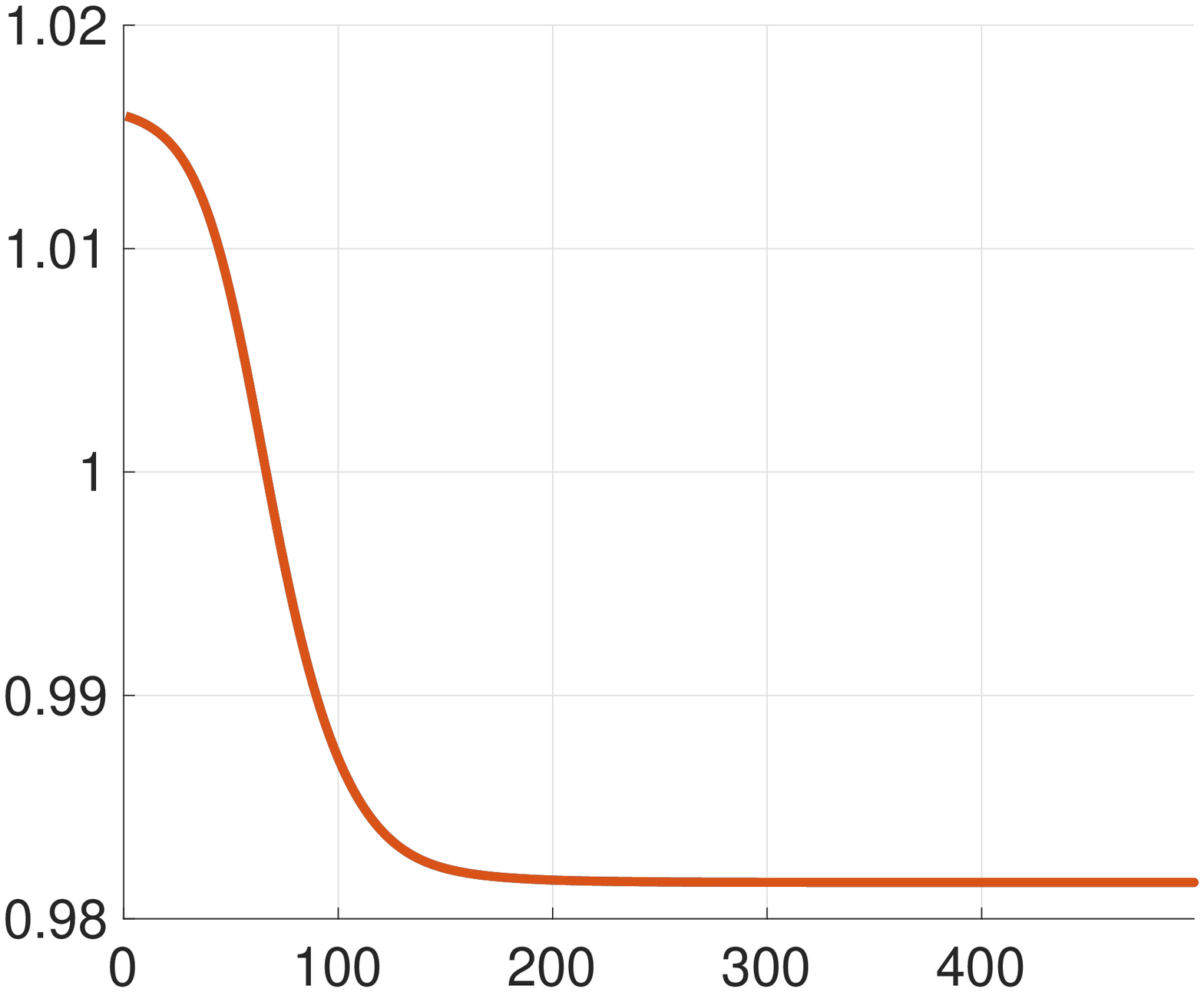}
     \put(-7,25){{\parbox{0.75\linewidth}\footnotesize \rotatebox{90}{\footnotesize$\rho(\Tilde{M}^2[t])$}
     }}\normalsize
     \put(50,0){\footnotesize 
    $t$
     }
  \end{overpic}
\caption{Spectral Radius of the state transition matrix of the Omicron variant in each country (left); Spectral Radius of the state transition matrix of the Delta variant in each country (right).}
\label{fig:Spectral_Radius}
\end{figure}

\begin{figure}
\centering
\begin{overpic}[width =  0.948\columnwidth]{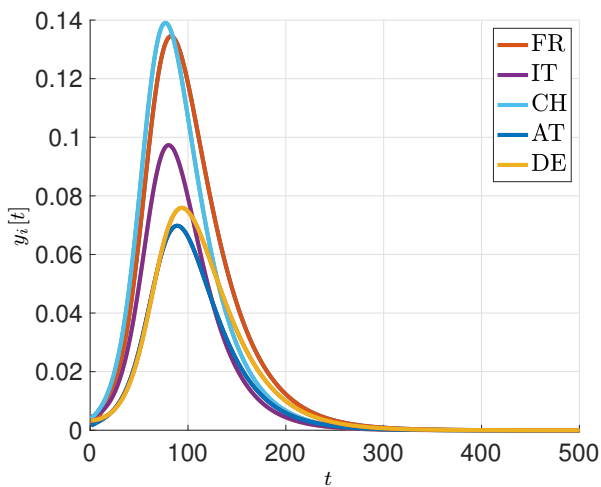}
     \put(0,38){{\parbox{0.75\linewidth}\footnotesize \rotatebox{90}{\footnotesize$y_i[t]$}
     }}\normalsize
     \put(50,0){\footnotesize 
    $t$
     }
   \end{overpic}
\caption{Total proportion of individuals who show similar symptoms from both variants in each country.}
\label{fig:y}
\end{figure}

The 
evolution of the infection levels
of both viruses are illustrated in Figure~\ref{fig:Evolution}, and the 
effective reproduction number, namely
$\rho(\Tilde{M}^k[t])$, of both viruses are plotted in Figure~\ref{fig:Spectral_Radius}. 
The infection level of each virus 
attains its 
peak at the same time when $\rho(\Tilde{M}^k[t])$ 
drops down to 1, consistent with our analysis in Section~\ref{sec:stability}. 
We
plot the total 
fraction
of individuals who 
exhibit
similar symptoms that are caused due to either of the viruses, namely $y_i[t]$, with $i=[5]$, from the observation model in Eq.~\eqref{eq:SIR_observation}; see
Figure~\ref{fig:y}.


We then estimate the infection level 
by using the 
Luenberger observer of Eq.~\eqref{eq:SIR_luenberger_observation}, with $\hat{s}[t]$ and $\hat{r}[t]$ 
calculated via
Eq.~\eqref{eq:hat:s} and Eq.~\eqref{eq:hat:r}, respectively. 
The initial conditions are provided in Table~\ref{table:initial_conditions}.
To investigate the impact of the 
observer gain $L_i^k$ on the estimation error, we simulate the Luenberger observer's performance with a scaled observer gain $\eta L_i^k$, where 
$L_i^k=1$, for all $i\in[5], k \in [2]$.
Let $t^*$ denote the time instant such that
the aggregated estimation error $\frac{1}{mn}\sum_{k=1}^m\sum_{i=1}^n |x_i^k[t]-\hat{x}_i^k[t]|<0.01$, for all $t>t^*$. 
In Figure~\ref{fig:Error_Gain}, we illustrate how 
the point in time at which the estimation error converges
($t^*$) depends on the gain of the observer ($\eta$) for this system.
We can see that as the observer gain increases, the convergence time first decreases and then increases; eventually, $t^*$ no longer exists 
because a large observer gain can cause the estimated infection level to exceed $1$, 
that is,
the estimated system 
states cease to be well-defined. 
\begin{figure}
\centering
\begin{overpic}[width = 0.7\columnwidth]{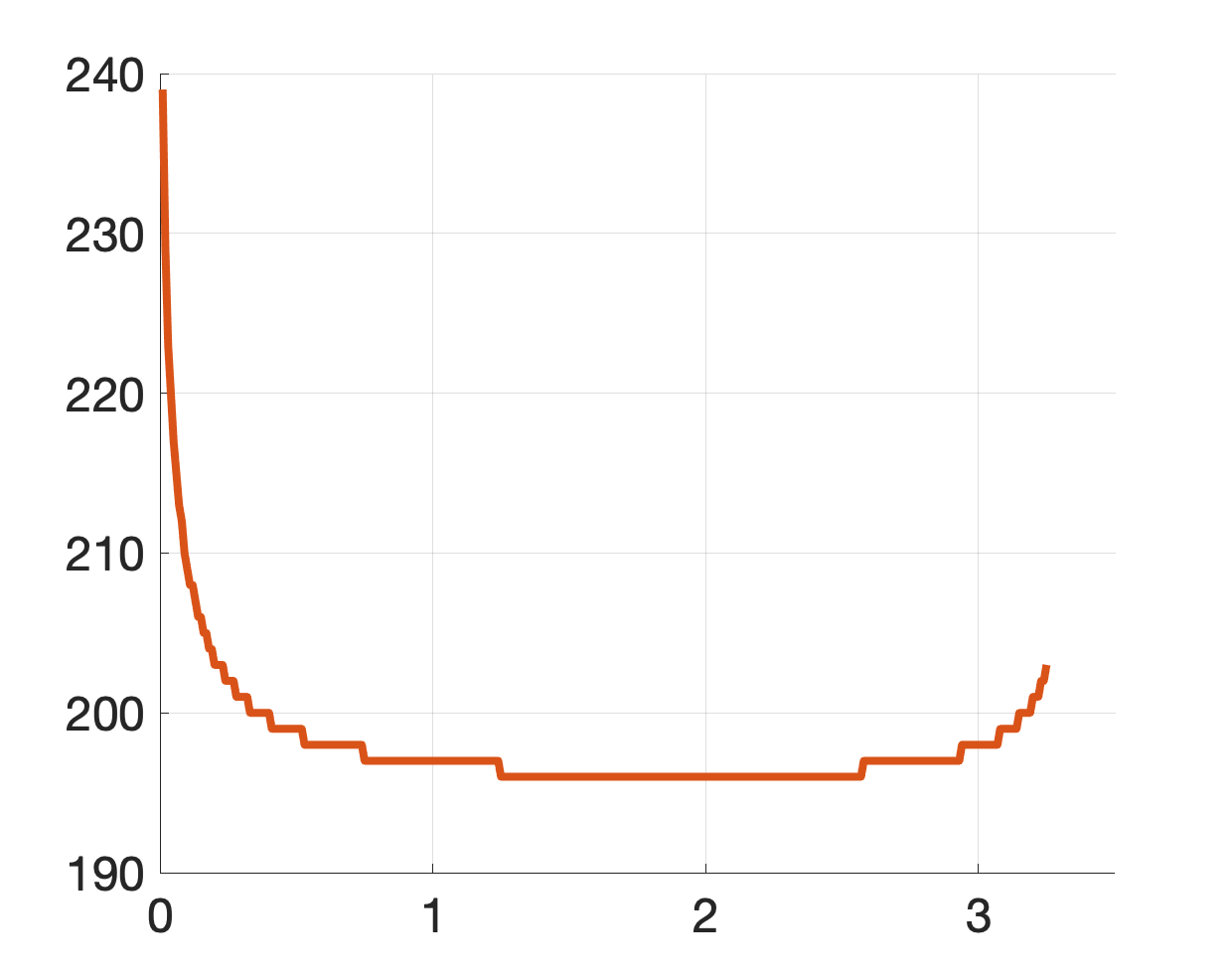}
\put(-1,38){{\parbox{0.75\linewidth}\small \rotatebox{90}{\small$t^*$}
     }}\normalsize
     \put(50,-1){\footnotesize{\parbox{0.75\linewidth}\footnotesize $\eta$
     }}\normalsize
  \end{overpic}
\caption{The scale of the observer gain vs the time for the estimation error to converge to zero.}
\label{fig:Error_Gain}
\end{figure}
In Figure~\ref{fig:Error_Gain}, when the scale of the observer gain $\eta>3.3$, 
the estimated infection level $\hat{x}_i^k[t] \rightarrow \infty$ as $t$ increases. 

We next 
focus on
finding the feasible observer gain $L^k$ such that the systems states are well-defined.
We 
set
$\tau^1=0.1$ and $\tau^2 = 0.3$, where $\tau^k$, for each $k \in [2]$, is as defined in Theorem~\ref{thm:Error_GAS}.
For different $l^k$ values, we calculate the corresponding $L^k$ values using LMI~\eqref{eq:lmi}. 
It turns out that setting $l^k  =10^{10}$ for both viruses  ensures that all the estimated system states are well-defined. Note that the feasible value of $l^k$ is not unique. For our subsequent simulations, we 
set  $l^k  =10^{10}$, and, by solving the LMI~\eqref{eq:lmi} in 
Corollary~\ref{coro:lmi}, we  obtain the observer gain $L^k$:
\begin{equation} \small
    L^1 = \begin{bmatrix}
    0.101223398722677\\
0.0928658303375023\\
0.112524328507691\\
0.0860241631907317\\
0.0843783515296357 
    \end{bmatrix},
        L^2 = \begin{bmatrix}
    0.0853417070451051\\
0.0879525432030855\\
0.0898592088737154\\
0.0885900881539504\\
0.0897704274804431 
    \end{bmatrix}. \label{eq:observer_gain_simulation}
\end{equation}
Notice that $\|L^2\|<\|L^1\|$. 
Thus, we need less compensation from the observer gain for the system state estimation of 
Delta than for that of Omicron.


We now explore whether we are able to retrieve the observer gains for all viruses through LMI~\eqref{eq:lmi}, when the assumptions in Theorem~\ref{thm:Observable_zero} are not satisfied. If we set $\gamma_{DE}^2 = \gamma_{DE}^1=0.2$ 
in Table~\ref{table:parameters_Flu}, then 
the LMI~\eqref{eq:lmi} becomes infeasible for the Delta variant, hence we are not able to accurately estimate the infection level for the Delta variant.
Therefore, when 
the healing rates of different viruses are identical at a node, we are unable to find the appropriate observer gain, which can be interpreted as saying that we are not able to estimate the system states of each virus if the viruses are indistinguishable (i.e., the recovery rates for the two viruses are the same) from each other at any subpopulation, consistent with Theorem~\ref{thm:Observable_zero}. 

\begin{table}[h!]
\centering
\begin{tabular}{|c | c c c c c|} 
 \hline
    & FR & IT & CH & AT & DE  \\ [0.5ex] 
 \hline
 $\hat{x}^1[0]$   & 0.0037& 0.0075& 0.0056& 0.0019& 0.0056 \\
 $\hat{x}^2[0]$ & 0.0005& 0.001& 0.002& 0.001& 0.0005  \\
 \hline
\end{tabular}
\caption{Initial conditions assumed for the Luenberger observer~\eqref{eq:SIR_luenberger_observation}.}
\label{table:initial_conditions}
\end{table}

\begin{figure}
\centering
\begin{overpic}[width = 0.948\columnwidth]{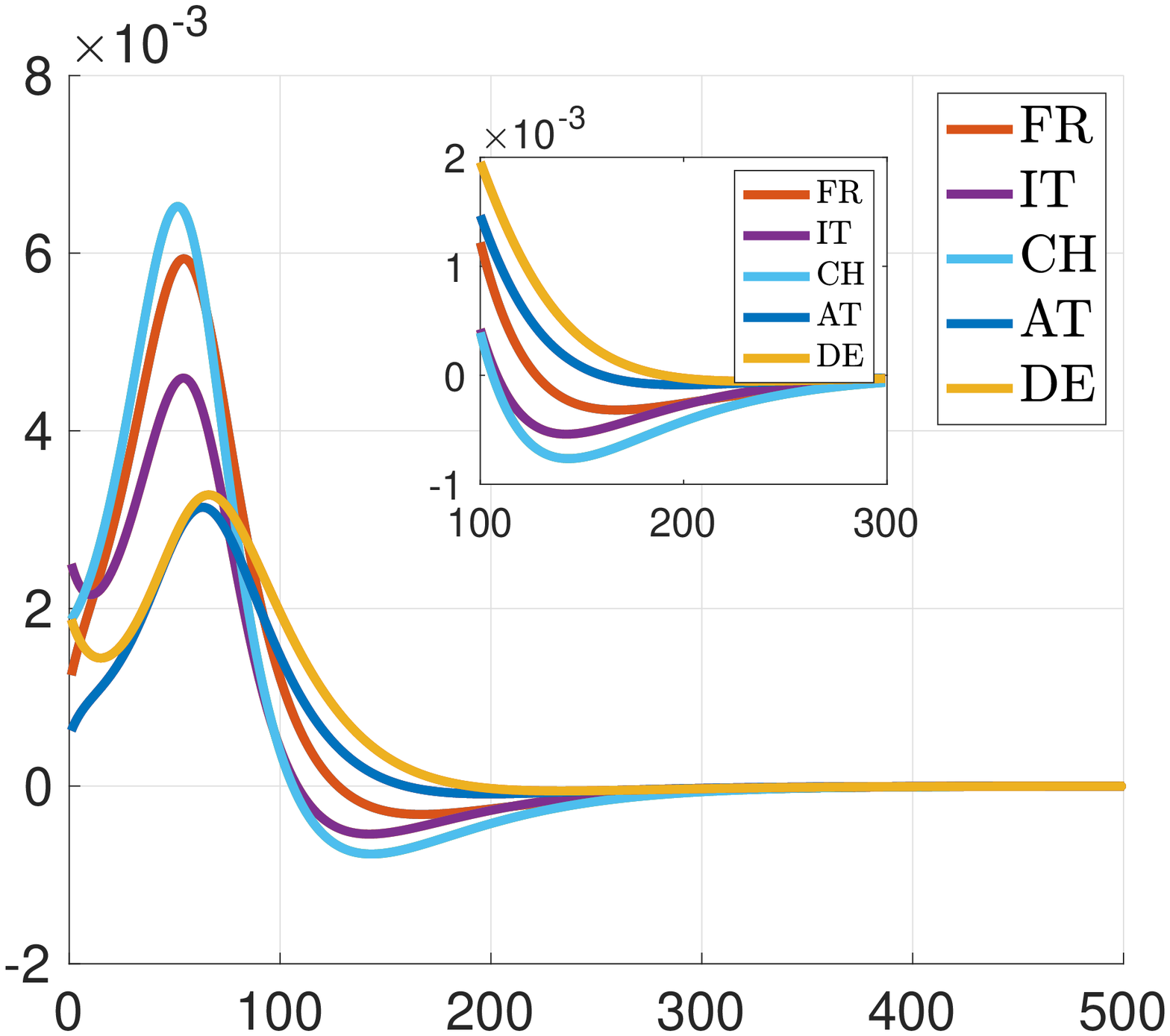}
     \put(2,35){{\parbox{0.75\linewidth}\footnotesize \rotatebox{90}{\footnotesize$x^1 -\hat{x}^1$}
     }}
     \put(50,1){\footnotesize{\parbox{0.75\linewidth}\footnotesize $t$
     }}\normalsize
   \end{overpic}
\begin{overpic}[width =  0.948\columnwidth]{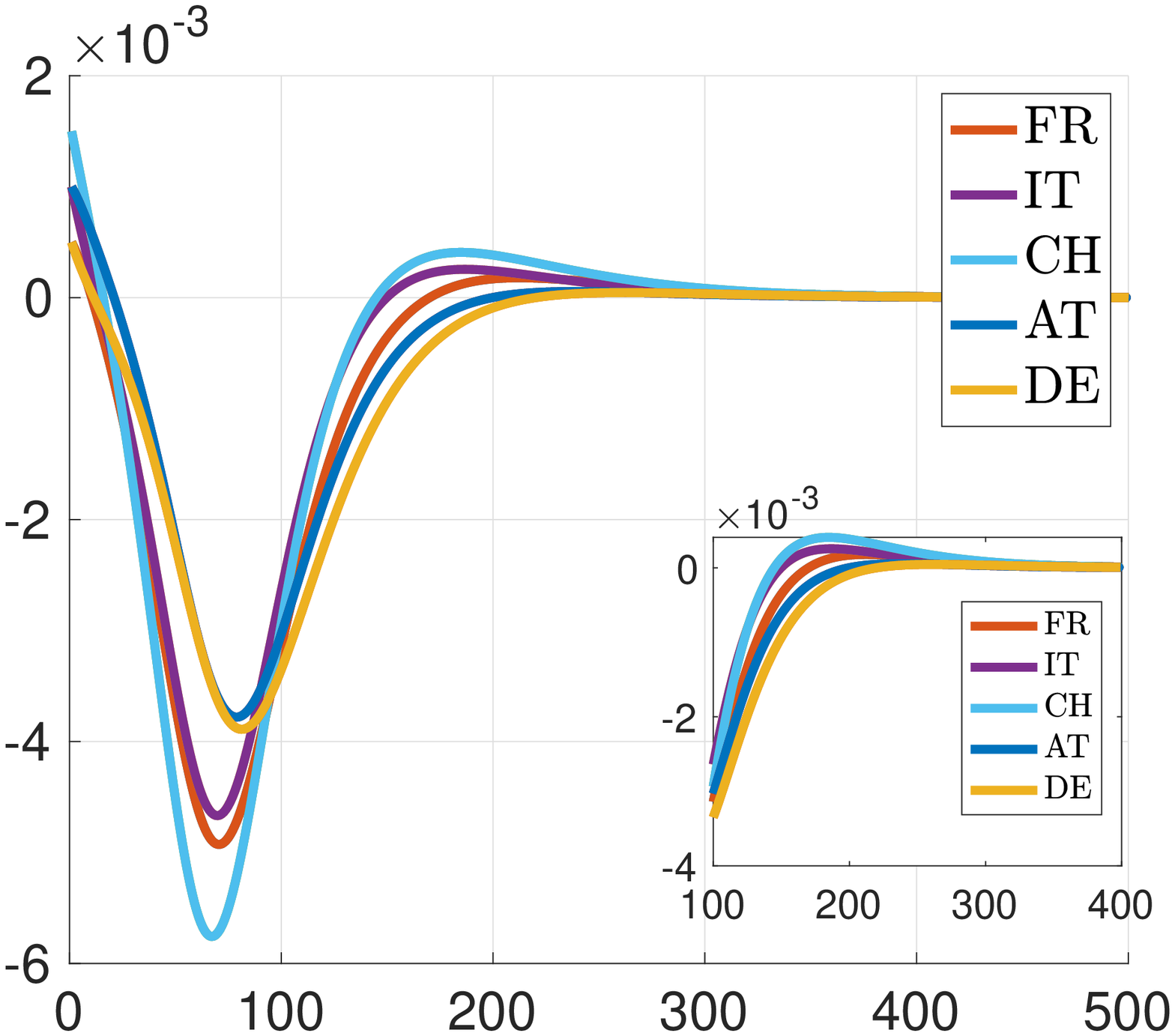}
     \put(2,35){{\parbox{0.75\linewidth}\footnotesize \rotatebox{90}{\footnotesize$x^2 -\hat{x}^2$}
     }}\normalsize
     \put(50,1){\footnotesize 
    $t$
     }
   \end{overpic}
\caption{Estimation error of infection level of the Omicron variant in each country (top); Estimation error of infection level of the Delta variant in each country (bottom).}
\label{fig:Error_nonoise}
\end{figure}

To evaluate the performance of our proposed estimation algorithm, we simulate the state estimation errors, with observer gains from~\eqref{eq:observer_gain_simulation} and initial conditions from Table~\ref{table:initial_conditions}, in Figure~\ref{fig:Error_nonoise}. 
We can see that,
in Figure~\ref{fig:Error_nonoise} (top),
the estimation error of the Omicron variant is negligible compared to its infection level, and converges to zero before its infection level does.
In Figure~\ref{fig:Error_nonoise} (bottom), the estimation error of the Delta variant starts to die out around $t=250$;
while the infection levels converge to zero when $t>300$, as seen in Figure~\ref{fig:Evolution} (right). 
Hence, the simulation results in Figure~\ref{fig:Error_nonoise} are consistent with 
Theorem~\ref{thm:Error_GAS}.

Inspired by Corollary~\ref{coro:lipschitz_const}, we now explore 
how to properly choose
the 
Lipschitz-like constant $l^k$ for our model by denoting:
\begin{equation}
    l^k_*[t] = \frac{||w^k[t]|| }{||x^k[t]-\hat{x}^k[t]||}
\end{equation}
for all $k \in [2]$. In Figure~\ref{fig:Lip_const}, we plot the Lipschitz-like constant $l^k_*[t]$ of our network (see Figure~\ref{fig:Graph_Europe}) for both viruses over time. 
Recall from Corollary~\ref{coro:lipschitz_const} that we must choose $l^k$ such that, for each $k  \in [2]$: \begin{equation}
    l^k \geq \frac{||w^k[t]|| }{||x^k[t]-\hat{x}^k[t]||}
\end{equation}
for all $t \in \mathbb{Z}_{\geq 0}$. 
Our choice of $l^k=10^{10}>l^k_*[t]$ for all $t \in \mathbb{Z}_{\geq 0}$, as seen in Figure~\ref{fig:Lip_const}, satisfies the requirement of Corollary~\ref{coro:lipschitz_const}.


\begin{figure}
\centering
\begin{overpic}[width = 0.48\columnwidth]{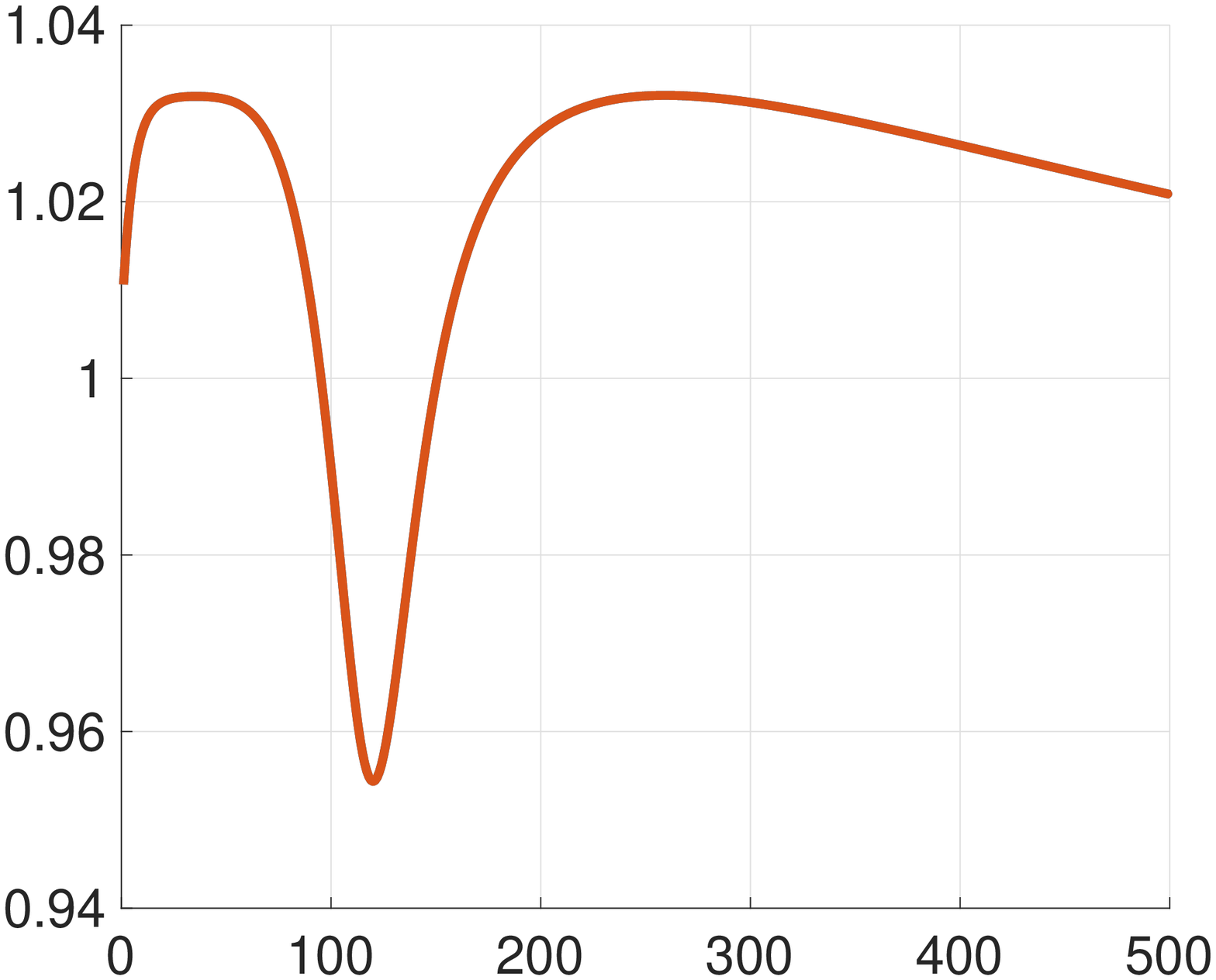}
     \put(-3,35){{\parbox{0.75\linewidth}\footnotesize \rotatebox{90}{\footnotesize$l^1_*$}
     }}
     \put(50,0){\footnotesize{\parbox{0.75\linewidth}\footnotesize $t$
     }}\normalsize
  \end{overpic}
\begin{overpic}[width =  0.48\columnwidth]{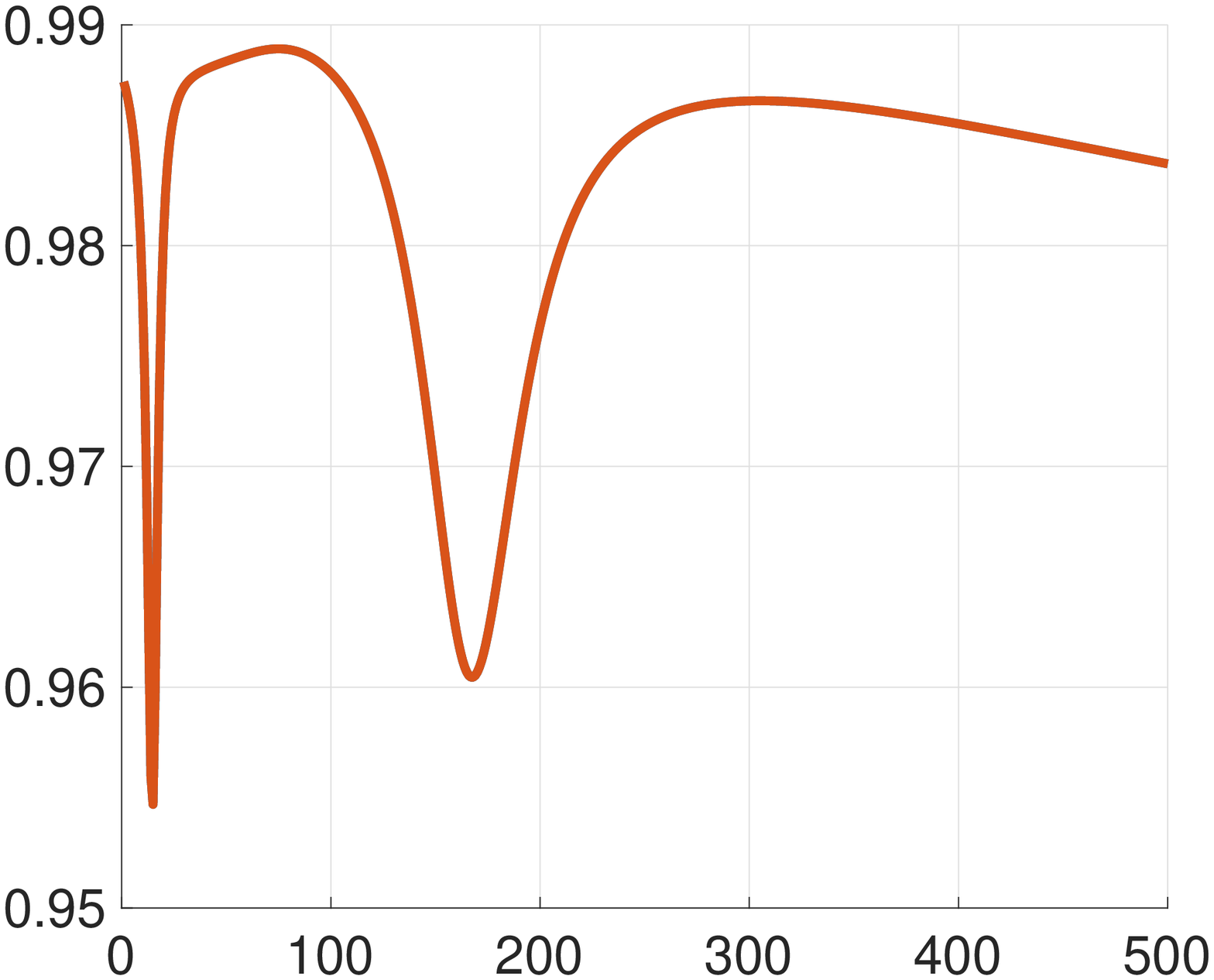}
     \put(-7,35){{\parbox{0.75\linewidth}\footnotesize \rotatebox{90}{\footnotesize$l^2_*$}
     }}\normalsize
     \put(50,0){\footnotesize 
    $t$
     }
  \end{overpic}
\caption{Lipschitz-like constant $l_*^k$ of the Omicron variant in each country (left); Lipschitz-like constant $l^k$ of the Delta variant in each country (right).}
\label{fig:Lip_const}
\end{figure}

\begin{figure}
\centering
\begin{overpic}[width = 0.48\columnwidth]{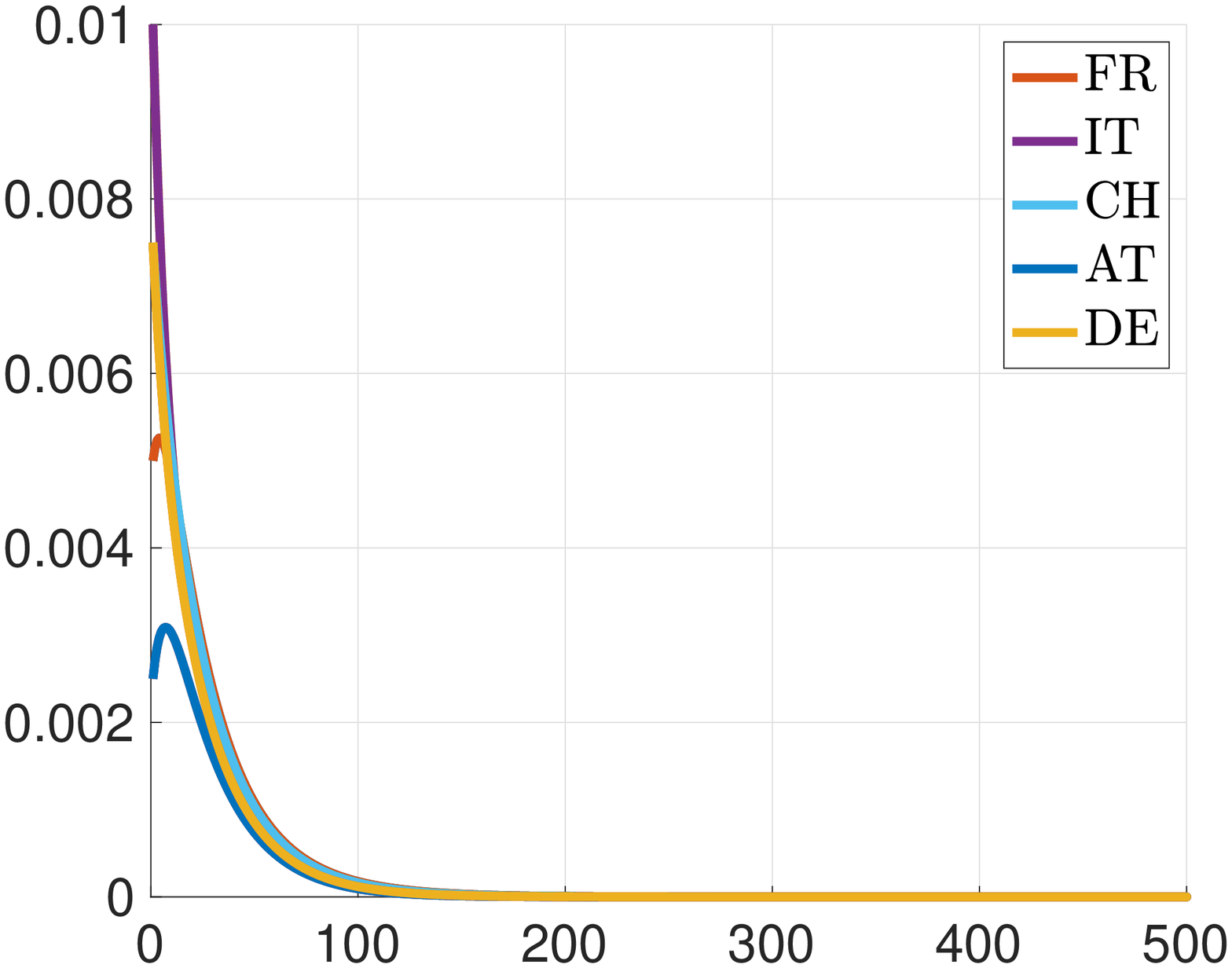}
     \put(-2,37){{\parbox{0.75\linewidth}\footnotesize \rotatebox{90}{\footnotesize$x^1$}
     }}
     \put(50,1){\footnotesize{\parbox{0.75\linewidth}\footnotesize $t$
     }}\normalsize
   \end{overpic}
   \vspace{1ex}
\begin{overpic}[width =  0.48\columnwidth]{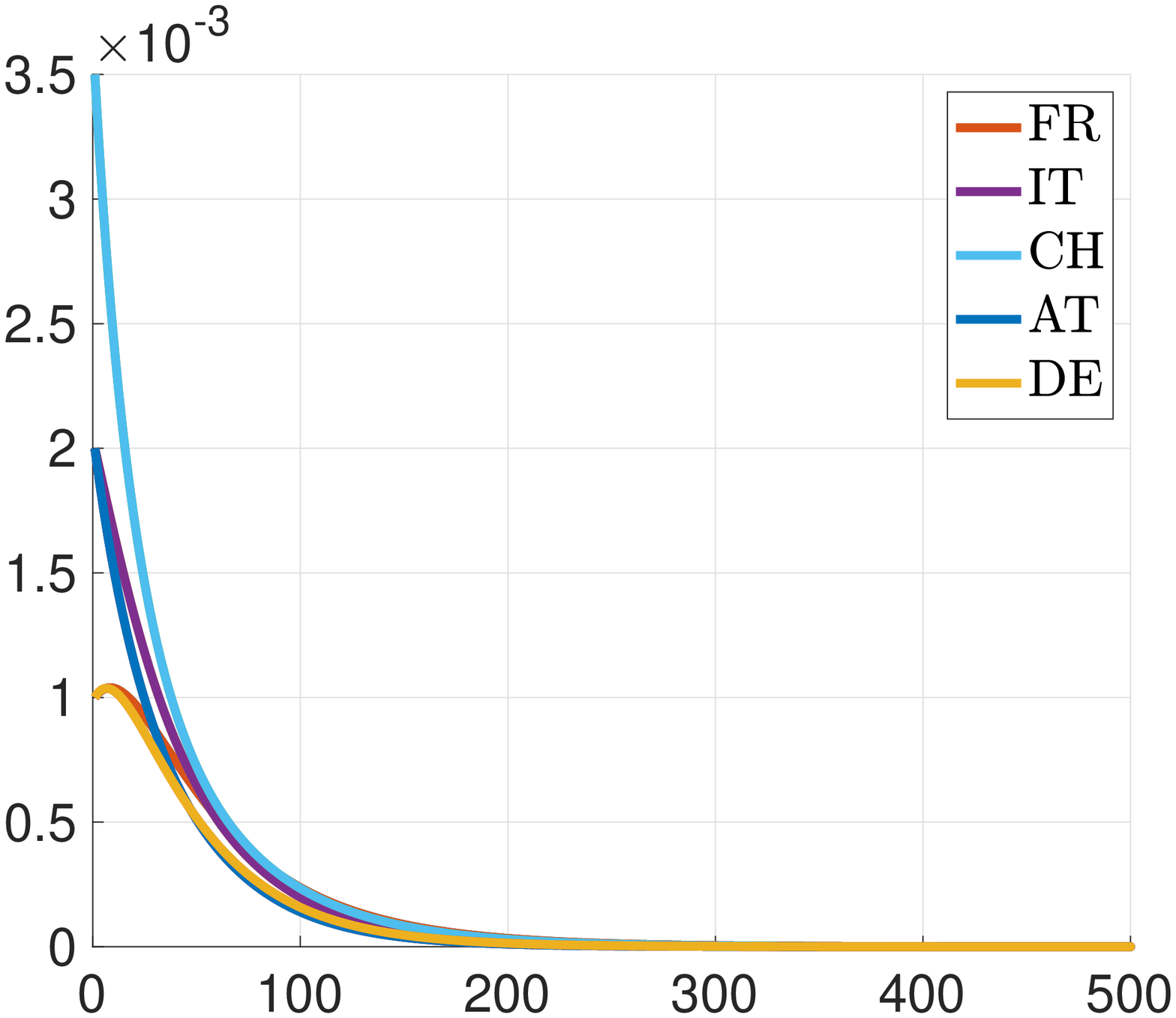}
     \put(-1,37){{\parbox{0.75\linewidth}\footnotesize \rotatebox{90}{\footnotesize$x^2 $}
     }}\normalsize
     \put(50,-1){\footnotesize 
    $t$
     }
   \end{overpic}
\caption{Evolution of infection level of the Omicron variant in each country with the distributed feedback control strategy in Eq.~\eqref{eq:control_scheme_feedback} (left); Evolution of infection level of the Delta variant in each country the distributed feedback control strategy in Eq.~\eqref{eq:control_scheme_feedback} (right).}
\label{fig:Distributed_Control}
\end{figure}

\begin{figure}
\centering
\begin{overpic}[width = 0.48\columnwidth]{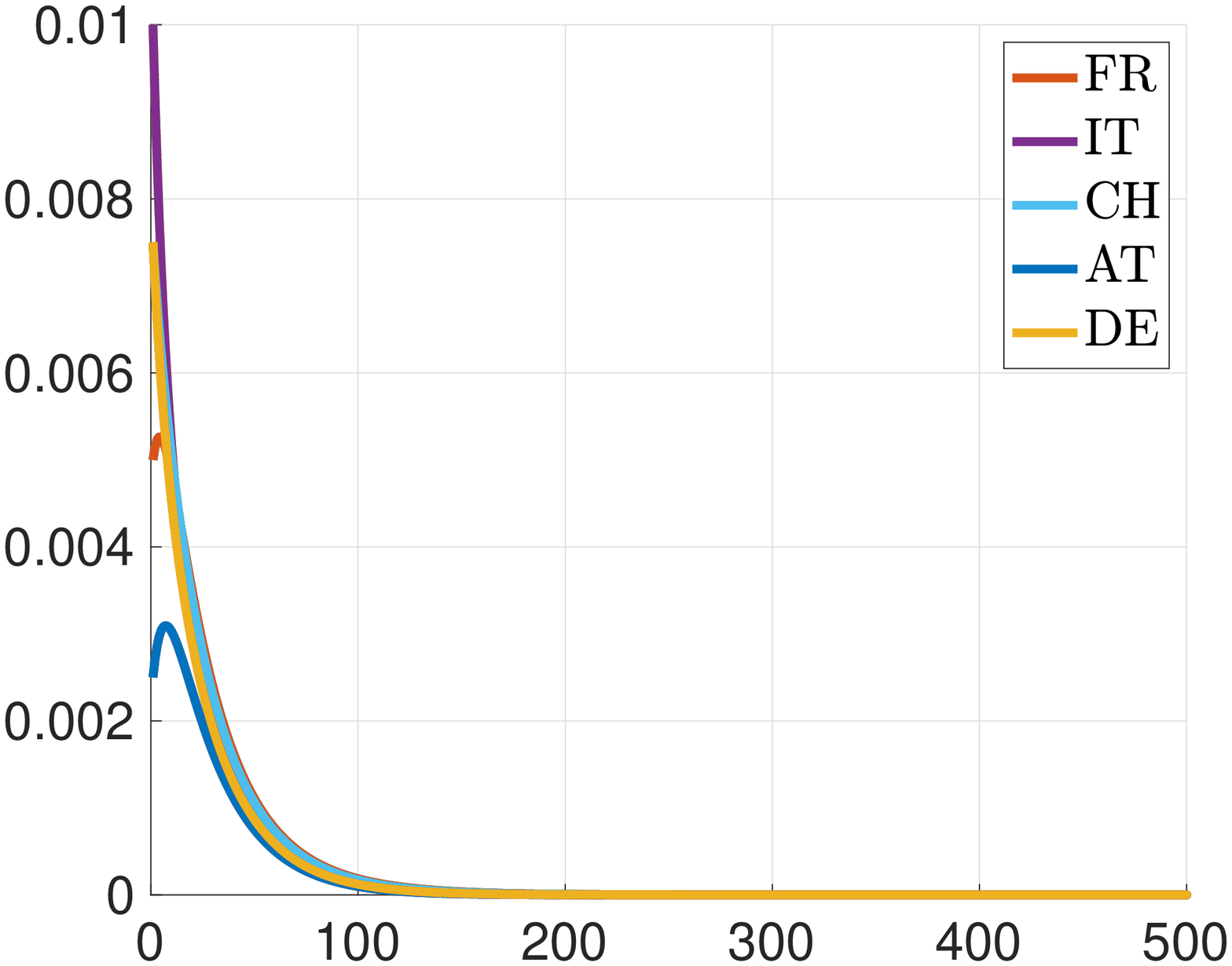}
     \put(-2,37){{\parbox{0.75\linewidth}\footnotesize \rotatebox{90}{\footnotesize$x^1$}
     }}
     \put(50,1){\footnotesize{\parbox{0.75\linewidth}\footnotesize $t$
     }}\normalsize
   \end{overpic}
   \vspace{1ex}
\begin{overpic}[width =  0.48\columnwidth]{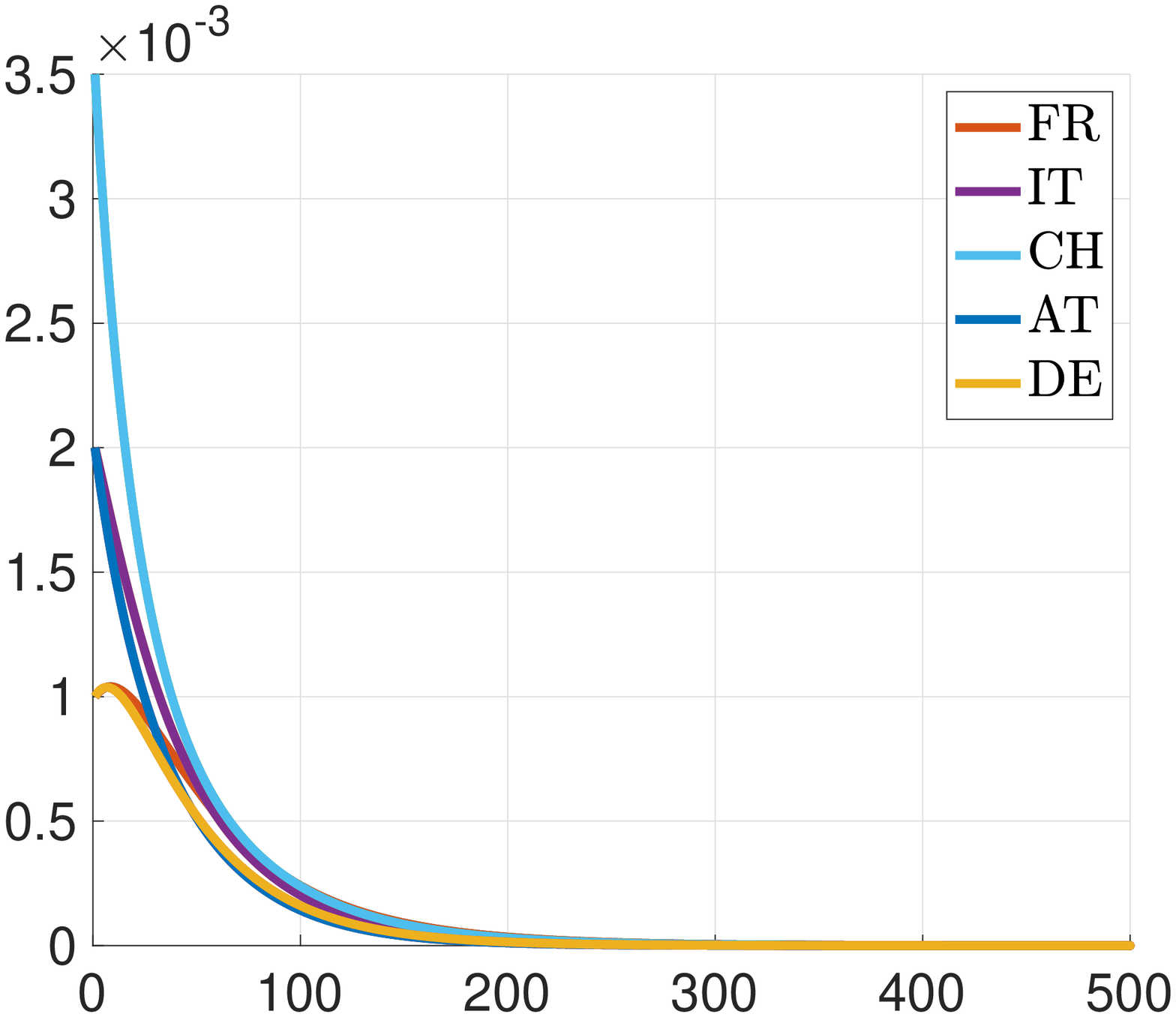}
     \put(-1,37){{\parbox{0.75\linewidth}\footnotesize \rotatebox{90}{\footnotesize$x^2 $}
     }}\normalsize
     \put(50,-1){\footnotesize 
    $t$
     }
   \end{overpic}
\caption{Evolution of the infection level of the Omicron variant in each country with the distributed mitigation strategy with a control input from Eq.~\eqref{eq:feedback_est} (left); Evolution of the infection level of the Delta variant in each country with the distributed mitigation strategy with a control input from Eq.~\eqref{eq:feedback_est} (right).}
\label{fig:Distributed_Control_est}
\end{figure}


In Figure~\ref{fig:Distributed_Control}, we implement the distributed feedback mitigation strategy of Eq.~\eqref{eq:control_scheme_feedback} over
all the countries in our network.
We can see that, 
consistent with Theorem~\ref{thm:control_feedback}, our proposed mitigation strategy is able to eradicate the spread of both viruses at an exponential rate. Moreover, maintaining the mitigation strategy over the network throughout time ensures that no subsequent infection wave occurs. 
We then explore the state feedback controller of Eq.~\eqref{eq:control_scheme_feedback} which utilizes the estimated system states instead of the actual system states, 
i.e., the control input is calculated as:
\begin{equation}
    \overline{u}^k_i[t] = -\hat{s}_i[t]\sum_{j=1}^n \beta_{ij}^k, \label{eq:feedback_est}
\end{equation}
where $\hat{s}_i[t]$ is obtained through the estimation algorithm in Eq.~\eqref{eq:SIR_luenberger_observation}. In Figure~\ref{fig:Distributed_Control_est}, we plot 
the infection level 
for both viruses using the mitigation strategy~\eqref{eq:control_scheme_feedback} with the control input~\eqref{eq:feedback_est} applied. It can be seen that the feedback controller with the estimated states is able to achieve the same goal, namely, eradicating each virus at an exponential rate. Even though 
Theorem~\ref{thm:control_feedback}
did not prove that the feedback controller that uses the estimated susceptible level Eq.~\eqref{eq:feedback_est} eradicates the viruses at an exponential rate, it seems that our estimation algorithm in 
Eq.~\eqref{eq:SIR_luenberger_observation}
is able to reconstruct the susceptible level 
accurately
in order to enable rapid eradication in this case.






\section{Conclusion}\label{sec:conclusion}
This paper has proposed a novel discrete-time networked multi-virus 
SIR
model.
We have 
provided two sufficient conditions for each virus to converge to zero exponentially. Then, we have specified a 
sufficient condition for the system to be strongly locally observable
when there are no susceptible individuals in the network.
We have proposed a distributed Luenberger
observer 
which estimates the system states. We have then provided 
a method for choosing the observer gain such that the estimation error of each virus converges to zero asymptotically.
We have introduced a distributed feedback mitigation strategy that ensured that each virus is eradicated at an exponential rate.
In the simulations, we have utilized the Luenberger state observer proposed to estimate the system states and the results illustrate that the Luenberger observer is suitable for 
state estimation 
of
our model. 

For future work, 
we plan to 
study the observability of the system
by relaxing the assumption that all nodes have zero susceptible individuals and allowing for more realistic and generic susceptible levels.
In addition, we would like to explore the system state estimation problem in the presence of measurement noise to emulate practical scenarios where there is error in the collected data.
We also aim to study more complex models, such as the SEIR and SAIR models, which incorporate additional compartments for exposed or asymptomatic individuals.





\ifCLASSOPTIONcaptionsoff
  \newpage
\fi



%

\normalem
\bibliographystyle{IEEEtran}
\bibliography{reference}

\begin{thebibliography}{10}
\providecommand{\url}[1]{#1}
\csname url@samestyle\endcsname
\providecommand{\newblock}{\relax}
\providecommand{\bibinfo}[2]{#2}
\providecommand{\BIBentrySTDinterwordspacing}{\spaceskip=0pt\relax}
\providecommand{\BIBentryALTinterwordstretchfactor}{4}
\providecommand{\BIBentryALTinterwordspacing}{\spaceskip=\fontdimen2\font plus
\BIBentryALTinterwordstretchfactor\fontdimen3\font minus
  \fontdimen4\font\relax}
\providecommand{\BIBforeignlanguage}[2]{{%
\expandafter\ifx\csname l@#1\endcsname\relax
\typeout{** WARNING: IEEEtran.bst: No hyphenation pattern has been}%
\typeout{** loaded for the language `#1'. Using the pattern for}%
\typeout{** the default language instead.}%
\else
\language=\csname l@#1\endcsname
\fi
#2}}
\providecommand{\BIBdecl}{\relax}
\BIBdecl

\bibitem{zhang2022networked}
C.~Zhang, S.~Gracy, T.~Ba{\c{s}}ar, and P.~E. Par{\'e}, ``A networked
  competitive multi-virus {SIR} model: Analysis and observability,''
  \emph{IFAC-PapersOnLine}, vol.~55, no.~13, pp. 13--18, 2022.

\bibitem{benedictow2004black}
O.~J. Benedictow and O.~L. Benedictow, \emph{The Black Death, 1346-1353: The
  Complete History}.\hskip 1em plus 0.5em minus 0.4em\relax Boydell \& Brewer,
  2004.

\bibitem{cheng2007happened}
K.~Cheng and P.~Leung, ``What happened in {C}hina during the 1918 influenza
  pandemic?'' \emph{International Journal of Infectious Diseases}, vol.~11,
  no.~4, pp. 360--364, 2007.

\bibitem{johnson2002updating}
N.~P. Johnson and J.~Mueller, ``Updating the accounts: Global mortality of the
  1918-1920 {S}panish influenza pandemic,'' \emph{Bulletin of the History of
  Medicine}, pp. 105--115, 2002.

\bibitem{whoswine}
{World Health Organization (WHO)}, ``2009 {H1N1 P}andemic,''
  \url{https://www.cdc.gov/flu/pandemic-resources/2009-h1n1-pandemic.html},
  accessed: 2022-12-08.

\bibitem{whoCoronavirus}
------, ``Global coronavirus {(2019-nCoV)},''
  \url{https://www.who.int/emergencies/diseases/novel-coronavirus-2019},
  accessed: 2022-12-08.

\bibitem{mei2017dynamics}
W.~Mei, S.~Mohagheghi, S.~Zampieri, and F.~Bullo, ``On the dynamics of
  deterministic epidemic propagation over networks,'' \emph{Annual Reviews in
  Control}, vol.~44, pp. 116--128, 2017.

\bibitem{weiss2013sir}
H.~H. Weiss, ``The {SIR} model and the foundations of public health,''
  \emph{Materials Mathematics}, pp. 0001--17, 2013.

\bibitem{colizza2007predictability}
V.~Colizza, A.~Barrat, M.~Barth{\'e}lemy, and A.~Vespignani, ``Predictability
  and epidemic pathways in global outbreaks of infectious diseases: the sars
  case study,'' \emph{BMC Medicine}, vol.~5, no.~1, pp. 1--13, 2007.

\bibitem{chang2017estimation}
H.-J. Chang, ``Estimation of basic reproduction number of the {Middle East}
  respiratory syndrome coronavirus ({MERS-CoV}) during the outbreak in {South
  Korea}, 2015,'' \emph{Biomedical Engineering Online}, vol.~16, no.~1, pp.
  1--11, 2017.

\bibitem{osthus2017forecasting}
D.~Osthus, K.~S. Hickmann, P.~C. Caragea, D.~Higdon, and S.~Y. Del~Valle,
  ``Forecasting seasonal influenza with a state-space {SIR} model,'' \emph{The
  Annals of Applied Statistics}, vol.~11, no.~1, p. 202, 2017.

\bibitem{coburn2009modeling}
B.~J. Coburn, B.~G. Wagner, and S.~Blower, ``Modeling influenza epidemics and
  pandemics: Insights into the future of swine flu {(H1N1)},'' \emph{BMC
  Medicine}, vol.~7, no.~1, pp. 1--8, 2009.

\bibitem{berge2017simple}
T.~Berge, J.-S. Lubuma, G.~Moremedi, N.~Morris, and R.~Kondera-Shava, ``A
  simple mathematical model for {Ebola in Africa},'' \emph{Journal of
  Biological Dynamics}, vol.~11, no.~1, pp. 42--74, 2017.

\bibitem{calafiore2020modified}
G.~C. Calafiore, C.~Novara, and C.~Possieri, ``A modified {SIR} model for the
  {COVID-19} contagion in {Italy},'' in \emph{Proc. 59th IEEE Conference on
  Decision and Control (CDC)}.\hskip 1em plus 0.5em minus 0.4em\relax IEEE,
  2020, pp. 3889--3894.

\bibitem{chen2020time}
Y.-C. Chen, P.-E. Lu, C.-S. Chang, and T.-H. Liu, ``A time-dependent {SIR}
  model for {COVID-19} with undetectable infected persons,'' \emph{IEEE
  Transactions on Network Science and Engineering}, vol.~7, no.~4, pp.
  3279--3294, 2020.

\bibitem{gracy2022modeling}
S.~Gracy, P.~E. Par{\'e}, J.~Liu, H.~Sandberg, C.~L. Beck, K.~H. Johansson, and
  T.~Ba{\c{s}}ar, ``Modeling and analysis of a coupled {SIS} bi-virus model,''
  \emph{Automatica, https://arxiv.org/pdf/2207.11414.pdf}, 2022, {N}ote: Under
  Review, $2^{nd}$ Round.

\bibitem{liu2019analysis}
J.~Liu, P.~E. Par{\'e}, A.~Nedi{\'c}, C.~Y. Tang, C.~L. Beck, and
  T.~Ba{\c{s}}ar, ``Analysis and control of a continuous-time bi-virus model,''
  \emph{IEEE Transactions on Automatic Control}, vol.~64, no.~12, pp.
  4891--4906, 2019.

\bibitem{ye2021convergence}
M.~Ye, B.~D. Anderson, and J.~Liu, ``Convergence and equilibria analysis of a
  networked bivirus epidemic model,'' \emph{SIAM Journal on Control and
  Optimization}, vol.~60, no.~2, pp. S323--S346, 2022.

\bibitem{pepin2008asymmetric}
K.~M. Pepin, K.~Lambeth, and K.~A. Hanley, ``Asymmetric competitive suppression
  between strains of dengue virus,'' \emph{BMC Microbiology}, vol.~8, no.~1,
  pp. 1--10, 2008.

\bibitem{nowak1991evolution}
M.~Nowak, ``The evolution of viruses. competition between horizontal and
  vertical transmission of mobile genes,'' \emph{Journal of Theoretical
  Biology}, vol. 150, no.~3, pp. 339--347, 1991.

\bibitem{poland1996two}
A.~M. Poland, H.~Vennema, J.~E. Foley, and N.~C. Pedersen, ``Two related
  strains of feline infectious peritonitis virus isolated from
  immunocompromised cats infected with a feline enteric coronavirus,''
  \emph{Journal of Clinical Microbiology}, vol.~34, no.~12, pp. 3180--3184,
  1996.

\bibitem{van2009virus}
P.~Van~Mieghem, J.~Omic, and R.~Kooij, ``Virus spread in networks,''
  \emph{IEEE/ACM Trans. on Net.}, vol.~17, no.~1, pp. 1--14, 2009.

\bibitem{pare2017multi}
P.~E. Par{\'e}, J.~Liu, C.~L. Beck, A.~Nedi{\'c}, and T.~Ba{\c{s}}ar,
  ``Multi-competitive viruses over static and time-varying networks,'' in
  \emph{Proceedings of the 2017 American Control Conference (ACC)}.\hskip 1em
  plus 0.5em minus 0.4em\relax IEEE, 2017, pp. 1685--1690.

\bibitem{prakash2012winner}
B.~A. Prakash, A.~Beutel, R.~Rosenfeld, and C.~Faloutsos, ``Winner takes all:
  Competing viruses or ideas on fair-play networks,'' in \emph{Proceedings of
  the 21st International Conference on World Wide Web}, 2012, pp. 1037--1046.

\bibitem{pare2020analysis}
P.~E. Par{\'e}, D.~Vrabac, H.~Sandberg, and K.~H. Johansson, ``Analysis, online
  estimation, and validation of a competing virus model,'' in \emph{Proceedings
  of the 2020 American Control Conference (ACC)}.\hskip 1em plus 0.5em minus
  0.4em\relax IEEE, 2020, pp. 2556--2561.

\bibitem{sahneh2014competitive}
F.~D. Sahneh and C.~Scoglio, ``Competitive epidemic spreading over arbitrary
  multilayer networks,'' \emph{Physical Review E}, vol.~89, no.~6, p. 062817,
  2014.

\bibitem{santos2015bi}
A.~Santos, J.~M. Moura, and J.~M. Xavier, ``Bi-virus {SIS} epidemics over
  networks: Qualitative analysis,'' \emph{IEEE Transactions on Network Science
  and Engineering}, vol.~2, no.~1, pp. 17--29, 2015.

\bibitem{liu2016analysis}
J.~Liu, P.~E. Par{\'e}, A.~Nedi{\'c}, C.~Y. Tang, C.~L. Beck, and
  T.~Ba{\c{s}}ar, ``On the analysis of a continuous-time bi-virus model,'' in
  \emph{Proceedings of the 2016 IEEE 55th Conference on Decision and Control
  (CDC)}.\hskip 1em plus 0.5em minus 0.4em\relax IEEE, 2016, pp. 290--295.

\bibitem{pare2021multi}
P.~E. Par{\'e}, J.~Liu, C.~L. Beck, A.~Nedi{\'c}, and T.~Ba{\c{s}}ar,
  ``Multi-competitive viruses over time-varying networks with mutations and
  human awareness,'' \emph{Automatica}, vol. 123, p. 109330, 2021.

\bibitem{janson2020networked}
A.~Janson, S.~Gracy, P.~E. Par{\'e}, H.~Sandberg, and K.~H. Johansson,
  ``Networked multi-virus spread with a shared resource: Analysis and
  mitigation strategies,'' \emph{arXiv preprint arXiv:2011.07569}, 2020.

\bibitem{zhang2021estimation}
C.~Zhang, H.~Leung, B.~Butler, and P.~Par{\'e}, ``Estimation and distributed
  eradication of {SIR} epidemics on networks,'' \emph{arXiv preprint
  arXiv:2102.12549}, 2021.

\bibitem{ito2022strict}
H.~Ito, ``Strict smooth {L}yapunov {F}unctions and vaccination control of the
  {SIR} model certified by {ISS},'' \emph{IEEE Transactions on Automatic
  Control}, vol.~67, no.~9, pp. 4514--4528, 2022.

\bibitem{she2021peak}
B.~She, H.~C. Leung, S.~Sundaram, and P.~E. Par{\'e}, ``Peak infection time for
  a networked {SIR} epidemic with opinion dynamics,'' in \emph{Proceedings of
  the 60th IEEE Conference on Decision and Control (CDC)}.\hskip 1em plus 0.5em
  minus 0.4em\relax IEEE, 2021, pp. 2104--2109.

\bibitem{she2022learning}
B.~She, S.~Sundaram, and P.~E. Par{\'e}, ``A learning-based model predictive
  control framework for real-time {SIR} epidemic mitigation,'' in
  \emph{Proceedings of the 2022 American Control Conference (ACC)}.\hskip 1em
  plus 0.5em minus 0.4em\relax IEEE, 2022, pp. 2565--2570.

\bibitem{smith2022convex}
K.~D. Smith and F.~Bullo, ``Convex optimization of the basic reproduction
  number,'' \emph{IEEE Transactions on Automatic Control}, 2022.

\bibitem{ignatov2022two}
A.~Ignatov and S.~Trigger, ``Two viruses competition in the {SIR} model of
  epidemic spread: Application to {COVID-19},'' \emph{MedRxiv preprint doi:
  https://doi. org/10.1101/2022.01}, vol.~11, 2022.

\bibitem{lopez2021effectiveness}
J.~Lopez~Bernal, N.~Andrews, C.~Gower, E.~Gallagher, R.~Simmons, S.~Thelwall,
  J.~Stowe, E.~Tessier, N.~Groves, G.~Dabrera \emph{et~al.}, ``Effectiveness of
  {COVID}-19 vaccines against the {B. 1.617. 2 (Delta)} variant,'' \emph{New
  England Journal of Medicine}, 2021.

\bibitem{barmparis2020estimating}
G.~D. Barmparis and G.~Tsironis, ``Estimating the infection horizon of
  {COVID-19} in eight countries with a data-driven approach,'' \emph{Chaos,
  Solitons \& Fractals}, vol. 135, p. 109842, 2020.

\bibitem{russell2020estimating}
T.~W. Russell, J.~Hellewell, C.~I. Jarvis, K.~Van~Zandvoort, S.~Abbott,
  R.~Ratnayake, S.~Flasche, R.~M. Eggo, W.~J. Edmunds, A.~J. Kucharski
  \emph{et~al.}, ``Estimating the infection and case fatality ratio for
  coronavirus disease {(COVID-19)} using age-adjusted data from the outbreak on
  the {D}iamond {P}rincess cruise ship, {F}ebruary 2020,''
  \emph{Eurosurveillance}, vol.~25, no.~12, p. 2000256, 2020.

\bibitem{meyerowitz2020systematic}
G.~Meyerowitz-Katz and L.~Merone, ``A systematic review and meta-analysis of
  published research data on {COVID-19} infection-fatality rates,''
  \emph{International Journal of Infectious Diseases}, 2020.

\bibitem{antonelli2022risk}
M.~Antonelli, J.~C. Pujol, T.~D. Spector, S.~Ourselin, and C.~J. Steves, ``Risk
  of long {COVID} associated with {D}elta versus {O}micron variants of
  {SARS-CoV-2},'' \emph{The Lancet}, vol. 399, no. 10343, pp. 2263--2264, 2022.

\bibitem{rader2022use}
B.~Rader, A.~Gertz, A.~D. Iuliano, M.~Gilmer, L.~Wronski, C.~M. Astley,
  K.~Sewalk, T.~J. Varrelman, J.~Cohen, R.~Parikh \emph{et~al.}, ``Use of
  at-home {COVID-19} tests—{U}nited {S}tates, {A}ugust 23, 2021--{M}arch 12,
  2022,'' \emph{Morbidity and Mortality Weekly Report}, vol.~71, no.~13, p.
  489, 2022.

\bibitem{perks2021covid}
R.~Perks and N.~Schneck, ``{COVID-19} in artisanal and small-scale mining
  communities: Preliminary results from a global rapid data collection
  exercise,'' \emph{Environmental Science \& Policy}, vol. 121, pp. 37--41,
  2021.

\bibitem{ma2021metagenomic}
S.~Ma, F.~Zhang, F.~Zhou, H.~Li, W.~Ge, R.~Gan, H.~Nie, B.~Li, Y.~Wang, M.~Wu
  \emph{et~al.}, ``Metagenomic analysis reveals oropharyngeal microbiota
  alterations in patients with {COVID-19},'' \emph{Signal Transduction and
  Targeted Therapy}, vol.~6, no.~1, pp. 1--11, 2021.

\bibitem{belongia2020covid}
E.~A. Belongia and M.~T. Osterholm, ``{COVID-19} and flu, a perfect storm,''
  pp. 1163--1163, 2020.

\bibitem{niazi2022observer}
M.~U.~B. Niazi and K.~H. Johansson, ``Observer design for the state estimation
  of epidemic processes,'' \emph{in Proceedings of the 61st IEEE Conference on
  Decision and Control (CDC)}, 2022.

\bibitem{bara2005observer}
G.~I. Bara, A.~Zemouche, and M.~Boutayeb, ``Observer synthesis for {L}ipschitz
  discrete-time systems,'' in \emph{Proceedings of the 2005 IEEE International
  Symposium on Circuits and Systems}.\hskip 1em plus 0.5em minus 0.4em\relax
  IEEE, 2005, pp. 3195--3198.

\bibitem{arcak2001nonlinear}
M.~Arcak and P.~Kokotovi{\'c}, ``Nonlinear observers: {A} circle criterion
  design and robustness analysis,'' \emph{Automatica}, vol.~37, no.~12, pp.
  1923--1930, 2001.

\bibitem{mitra2018distributed}
A.~Mitra and S.~Sundaram, ``Distributed observers for {LTI} systems,''
  \emph{IEEE Transactions on Automatic Control}, vol.~63, no.~11, pp.
  3689--3704, 2018.

\bibitem{snow1855mode}
J.~Snow, \emph{On the Mode of Communication of Cholera}.\hskip 1em plus 0.5em
  minus 0.4em\relax John Churchill, 1855.

\bibitem{whoEbola}
WHO, ``Ebola virus disease – {Democratic Republic of the Congo},''
  \url{https://www.who.int/csr/don/30-january-2020-ebola-drc/en/}, accessed:
  2020-02-03.

\bibitem{atkinson2008introduction}
K.~E. Atkinson, \emph{An {I}ntroduction to {N}umerical {A}nalysis}.\hskip 1em
  plus 0.5em minus 0.4em\relax John Wiley \& sons, 2008.

\bibitem{brauer2019mathematical}
F.~Brauer, C.~Castillo-Chavez, and Z.~Feng, \emph{Mathematical {M}odels in
  {E}pidemiology}.\hskip 1em plus 0.5em minus 0.4em\relax Springer, 2019,
  vol.~32.

\bibitem{teunis2008norwalk}
P.~F. Teunis, C.~L. Moe, P.~Liu, S.~E.~Miller, L.~Lindesmith, R.~S. Baric,
  J.~Le~Pendu, and R.~L. Calderon, ``Norwalk virus: how infectious is it?''
  \emph{Journal of Medical Virology}, vol.~80, no.~8, pp. 1468--1476, 2008.

\bibitem{panovska2020determining}
J.~Panovska-Griffiths, C.~C. Kerr, R.~M. Stuart, D.~Mistry, D.~J. Klein, R.~M.
  Viner, and C.~Bonell, ``Determining the optimal strategy for reopening
  schools, the impact of test and trace interventions, and the risk of
  occurrence of a second {COVID-19} epidemic wave in the {UK}: {A} modelling
  study,'' \emph{The Lancet Child \& Adolescent Health}, vol.~4, no.~11, pp.
  817--827, 2020.

\bibitem{vidyasagar2002nonlinear}
M.~Vidyasagar, \emph{Nonlinear {Systems Analysis}}.\hskip 1em plus 0.5em minus
  0.4em\relax SIAM, 2002.

\bibitem{rugh1996linear}
W.~J. Rugh, \emph{Linear {System Theory}}.\hskip 1em plus 0.5em minus
  0.4em\relax Prentice Hall Upper Saddle River, NJ, 1996, vol.~2.

\bibitem{rantzer2011distributed}
A.~Rantzer, ``Distributed control of positive systems,'' in \emph{Proceedings
  of the 50th IEEE Conference on Decision and Control and European Control
  Conference}, 2011, pp. 6608--6611.

\bibitem{sontag1979observability}
E.~D. Sontag, ``On the observability of polynomial systems, {I: F}inite-time
  problems,'' \emph{SIAM Journal on Control and Optimization}, vol.~17, no.~1,
  pp. 139--151, 1979.

\bibitem{nijmeijer1982observability}
H.~Nijmeijer, ``Observability of autonomous discrete time non-linear systems: a
  geometric approach,'' \emph{International Journal of Control}, vol.~36,
  no.~5, pp. 867--874, 1982.

\bibitem{gracy2020analysis}
S.~Gracy, P.~E. Par{\'e}, H.~Sandberg, and K.~H. Johansson, ``Analysis and
  distributed control of periodic epidemic processes,'' \emph{IEEE Transactions
  on Control of Network Systems}, vol.~8, no.~1, pp. 123--134, 2020.

\bibitem{carothers2000real}
N.~L. Carothers, \emph{Real {A}nalysis}.\hskip 1em plus 0.5em minus 0.4em\relax
  Cambridge University Press, 2000.

\bibitem{whitley2001herpes}
R.~J. Whitley and B.~Roizman, ``Herpes simplex virus infections,'' \emph{The
  Lancet}, vol. 357, no. 9267, pp. 1513--1518, 2001.

\bibitem{zhang2006schur}
F.~Zhang, \emph{The Schur {C}omplement and its {A}pplications}.\hskip 1em plus
  0.5em minus 0.4em\relax Springer Science \& Business Media, 2006, vol.~4.

\bibitem{ye2021applications}
M.~Ye, J.~Liu, B.~D. Anderson, and M.~Cao, ``Applications of the
  {P}oincar{\'e}--{H}opf {T}heorem: {E}pidemic {M}odels and {L}otka--{V}olterra
  {S}ystems,'' \emph{IEEE Transactions on Automatic Control}, vol.~67, no.~4,
  pp. 1609--1624, 2021.

\bibitem{zipfel2021missing}
C.~M. Zipfel, V.~Colizza, and S.~Bansal, ``The missing season: The impacts of
  the {COVID}-19 pandemic on influenza,'' \emph{Vaccine}, vol.~39, no.~28, pp.
  3645--3648, 2021.

\bibitem{shrestha2022evolution}
L.~B. Shrestha, C.~Foster, W.~Rawlinson, N.~Tedla, and R.~A. Bull, ``Evolution
  of the {SARS-CoV-2} {O}micron variants {BA}. 1 to {BA}. 5: Implications for
  immune escape and transmission,'' \emph{Reviews in Medical Virology},
  vol.~32, no.~5, p. e2381, 2022.

\bibitem{bolze2022evidence}
A.~Bolze, T.~Basler, S.~White, A.~Dei~Rossi, D.~Wyman, H.~Dai, P.~Roychoudhury,
  A.~L. Greninger, K.~Hayashibara, M.~Beatty \emph{et~al.}, ``Evidence for
  {SARS-CoV-2 Delta and Omicron} co-infections and recombination,'' \emph{Med},
  vol.~3, no.~12, pp. 848--859, 2022.

\end{thebibliography}

%








\begin{IEEEbiography}[{\includegraphics[width=1in,height=1.25in,
clip,keepaspectratio]{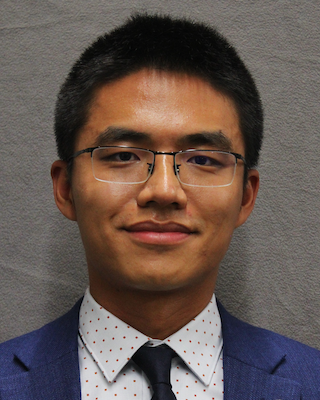}}]{Ciyuan Zhang} is a Ph.D. student in the School of Electrical and Computer Engineering at Purdue University. He received his M.S. in Electrical Engineering from Columbia University in 2018 and his B.S. in Electrical Engineering from Xi'an Jiaotong University in 2016. He is a recipient of Purdue University's Bilsland Dissertation Fellowship. His research interests include estimating and controlling epidemic dynamics on networked systems.
\end{IEEEbiography} 

\begin{IEEEbiography}[{\includegraphics[width=1in,height=1.25in, trim = 450 200 250 0,
clip,keepaspectratio]{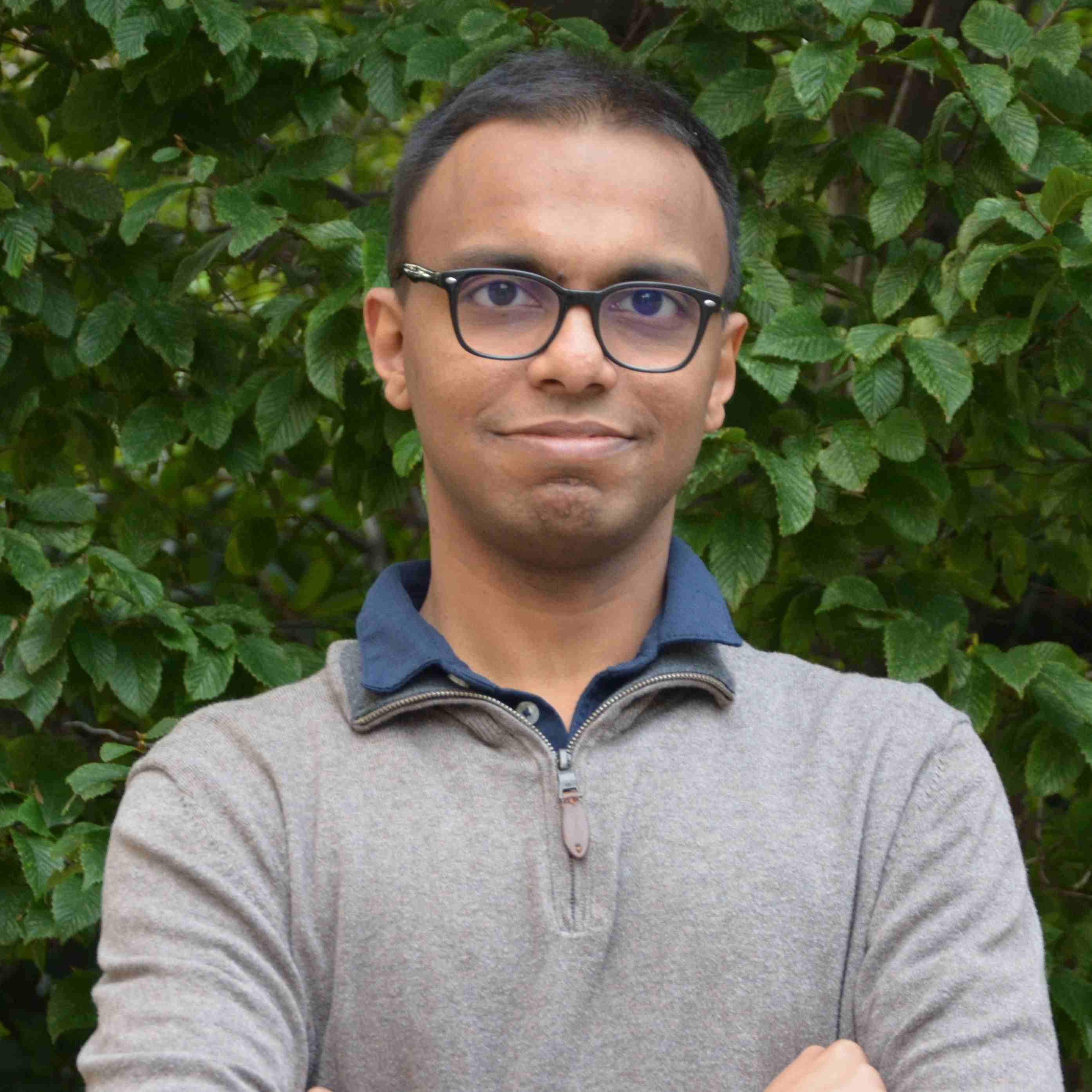}}]{Sebin Gracy} is a 
Post-Doctoral Associate in the Department of Electrical and Computer Engineering at Rice University, Houston, TX. Previously, he was a Post-Doctoral Researcher in the Division of Decision and Control Systems in the School of Electrical Engineering and Computer Science at KTH Royal Institute of Technology. He obtained his Ph.D. degree at Universit\'e Grenoble-Alpes in November, 2018. Prior to that, he obtained his M.S. and B.E. degrees in Electrical Engineering from the University of Colorado at Boulder and the University of Mumbai, in December, 2013 and June 2010, respectively. His research focuses on networked dynamical systems, with a particular emphasis on the theory of spreading processes.
\end{IEEEbiography} 

\begin{IEEEbiography} [{\includegraphics[width=1in,height=1.25in,clip,keepaspectratio]{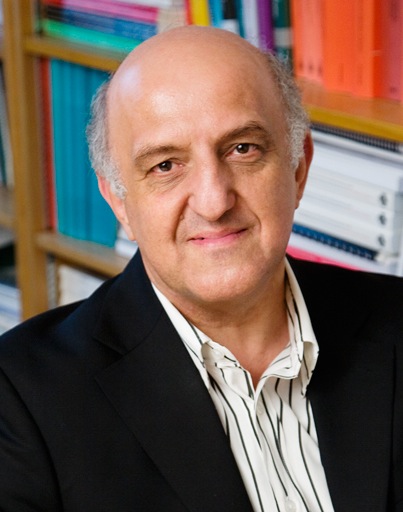}}]
{\bf Tamer Ba\c{s}ar} (S'71-M'73-SM'79-F'83-LF'13) is with the University of Illinois at Urbana-Champaign (UIUC), where he holds the academic positions of
Swanlund Endowed Chair Emeritus;
Center for Advanced Study Professor Emeritus of  Electrical and Computer Engineering;
Research Professor at the Coordinated Science
Laboratory; and Research Professor  at the Information Trust Institute.
He received B.S.E.E. from Robert College, Istanbul,
and M.S., M.Phil, and Ph.D. from Yale University. He is a member of the US National Academy
of Engineering, a member of the American Academy of Arts and Sciences, and Fellow of IEEE, IFAC and SIAM, and has served as president of IEEE CSS,
ISDG, and AACC. He has received several awards and recognitions over the years, including the
highest awards of IEEE CSS, IFAC, AACC, and ISDG; the IEEE Control Systems Award; and a number of international honorary
doctorates and professorships. He has over 1000 publications in systems, control, communications,
and dynamic games, including books on non-cooperative dynamic game theory, robust control,
network security, wireless and communication networks, and stochastic networked control. He was
the Editor-in-Chief of Automatica between 2004 and 2014, and is currently  editor of several book series. His current research interests
include stochastic teams, games, and networks; multi-agent systems and learning; spread of misinformation and epidemics;  security; and cyber-physical systems.
\end{IEEEbiography}

\begin{IEEEbiography}[{\includegraphics[width=1in,height=1.25in,
clip,keepaspectratio]{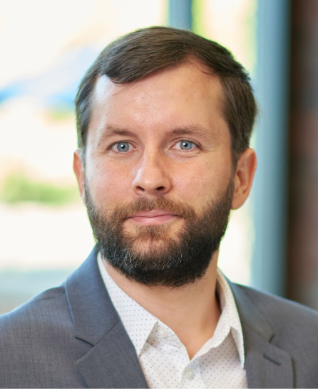}}]{Philip E. Par\'{e}} (Member, IEEE) is an Assistant Professor in the School of Electrical and Computer Engineering at Purdue University. He received his Ph.D. in Electrical and Computer Engineering (ECE) from the University of Illinois at Urbana-Champaign (UIUC) in 2018, after which he went to KTH Royal Institute of Technology in Stockholm, Sweden to be a Post-Doctoral Scholar. 
He received his B.S. in Mathematics with University Honors and his M.S. in Computer Science from Brigham Young University in 2012 and 2014, respectively. 
He is a recipient of the NSF CAREER award and was an inaugural Societal Impact Fellow at Purdue. His research focuses on networked control systems, namely modeling, analysis, and control of virus spread over networks.
\end{IEEEbiography}

\end{document}